%% file: main.tex
\documentclass{article}
\usepackage{ifthen}
\usepackage{mdwlist}
\usepackage{amsmath,amssymb,amsfonts,amsthm}
\usepackage{bm}
\usepackage{hyperref}
\usepackage{enumitem}
\usepackage{graphicx}
\usepackage{xspace}
\usepackage{verbatim}
\usepackage{algorithm}
\usepackage{algpseudocode}
\usepackage[margin=1in]{geometry}
\usepackage{color}
\usepackage{thm-restate}
\usepackage{latexsym}
\usepackage{epsfig}
\usepackage{pgf}
\usepackage{tikz}
\usetikzlibrary{tikzmark}
\usepackage{xcolor}
\usepackage{colortbl}

\input{preamble}

\title{Which Distribution Distances are Sublinearly Testable?}

\author {
Constantinos Daskalakis\thanks{Supported by NSF CCF-1617730, CCF-1650733, and ONR N00014-12-1-0999.}\\
EECS \& CSAIL, MIT\\
\tt{costis@csail.mit.edu}
\and
Gautam Kamath\thanks{Supported by NSF CCF-1617730, CCF-1650733, and ONR N00014-12-1-0999. Part of this work was done while the author was an intern at Microsoft Research New England.} \\
EECS \& CSAIL, MIT\\
\tt{g@csail.mit.edu}
\and
John Wright\thanks{Supported by NSF grant CCF-6931885.} \\
Physics, MIT\\
\tt{jswright@mit.edu}
}
\begin{document}
\maketitle
\begin{abstract}
\input{abs}
\end{abstract}

\input{intro}
\input{preliminaries}
\input{ones-ub}
\input{twos-ub}
\input{est-ub}
\input{lb}
\bibliographystyle{alpha}
\bibliography{biblio}
\appendix
\input{distanceinequalities}
\end{document}

%% file: preamble.tex
\def\d{\delta}

\renewcommand{\epsilon}{\ve}
\def\ve{\varepsilon}

\newcommand{\E}{\mbox{\bf E}}
\newcommand{\Var}{\mbox{\bf Var}}

\newcommand{\pr}[2][]{\mbox{Pr}\ifthenelse{\not\equal{}{#1}}{_{#1}}{}\!\left[#2\right]}

\newcommand{\dtv}{d_{\mathrm {TV}}}

\newcommand{\dkl}{d_{\mathrm {KL}}}
\newcommand{\dxs}{d_{\chi^2}}

\renewcommand{\dh}{d_{\mathrm {H}}}

\newcommand{\dlt}{d_{\ell_2}}

\newcommand{\calC}{\mathcal{C}}
\newcommand{\calF}{\mathcal{F}}
\newcommand{\bX}{\boldsymbol{X}}
\newcommand{\bY}{\boldsymbol{Y}}
\newcommand{\bZ}{\boldsymbol{Z}}

\newtheorem{theorem}{Theorem}

\newtheorem{proposition}{Proposition}

\newtheorem{remark}{Remark}
\newtheorem{fact}{Fact}
\newtheorem{lemma}{Lemma}

\newtheorem{corollary}{Corollary}

\newtheorem{definition}{Definition}

\newcommand{\ignore}[1]{}
\providecommand{\poly}{\operatorname*{poly}}

\newcommand{\bg}[1]{\medskip\noindent{\bf #1}}

\newcommand{\accept}{\textsc{Accept}\xspace}
\newcommand{\reject}{\textsc{Reject}\xspace}

\definecolor{Red}{rgb}{1,0,0}

\definecolor{Blue}{rgb}{0,0,1}

\newcommand{\oldbound}[1]{{}}

%% file: abs.tex
Given samples from an unknown distribution $p$ and a description of a distribution $q$, are $p$ and $q$ close or far?
This question of ``identity testing'' has received significant attention in the case of testing whether $p$ and $q$ are equal or far in total variation distance.
However, in recent work~\cite{ValiantV11a, AcharyaDK15, DaskalakisP17}, the following questions have been been critical to solving problems at the frontiers of distribution testing:
\begin{itemize}
\item Alternative Distances: Can we test whether $p$ and $q$ are far in other distances, say Hellinger?
\item Tolerance: Can we test when $p$ and $q$ are \emph{close}, rather than equal? And if so, close in which distances?
\end{itemize}

Motivated by these questions, we characterize the complexity of distribution testing under a variety of distances, including total variation, $\ell_2$, Hellinger, Kullback-Leibler, and $\chi^2$.
For each pair of distances $d_1$ and $d_2$, we study the complexity of testing if $p$ and $q$ are close in $d_1$ versus far in $d_2$, with a focus on identifying which problems allow \emph{strongly} sublinear testers (i.e., those with complexity $O(n^{1 - \gamma})$ for some $\gamma > 0$ where $n$ is the size of the support of the distributions $p$ and $q$).
We provide matching upper and lower bounds for each case.
We also study these questions in the case where we only have samples from $q$ (equivalence testing), showing qualitative differences from identity testing in terms of when tolerance can be achieved.
Our algorithms fall into the classical paradigm of $\chi^2$-statistics, but require crucial changes to handle the challenges introduced by each distance we consider.
Finally, we survey other recent results in an attempt to serve as a reference for the complexity of various distribution testing problems.

%% file: intro.tex
\section{Introduction} \label{sec:intro}

The arch problem in science is determining whether observations of some phenomenon conform to a conjectured model. Often, phenomena of interest are probabilistic in nature, and so are our models of these phenomena; hence, testing their validity becomes a statistical hypothesis testing problem. In mathematical notation, suppose that we have access to samples from some unknown distribution $p$ over some set $\Sigma$ of size $n$. We also have a hypothesis distribution $q$, and our goal is to distinguish whether $p=q$ or $p \neq q$. For instance, we may want to test whether the sizes of some population of insects are normally distributed around their mean by sampling insects and measuring their sizes.

Of course, our models are usually imperfect. In our insect example, perhaps our estimation of the mean and variance of the insect sizes is a bit off. Furthermore, the sizes will clearly always be positive numbers. Yet a Normal distribution could still be a good fit. To get a meaningful testing problem some slack may be introduced, turning the problem into that of distinguishing whether $d_1(p,q) \le \epsilon_1$ versus $d_2(p,q)\ge\epsilon_2$, for some distance measures $d_1(\cdot,\cdot)$ and $d_2(\cdot,\cdot)$ between distributions over $\Sigma$ and some choice of $\epsilon_1$ and $\epsilon_2$ which may potentially depend on $\Sigma$ or even $q$. Regardless, for the problem to be well-defined, the sets of distributions ${\cal C} = \{p~|~d_1(p,q) \le \epsilon_1 \}$ and ${\cal F} = \{p~|~d_2(p,q) \ge \epsilon_2 \}$ should be disjoint. In fact, as our goal is to distinguish between $p \in {\cal C}$ and $p \in {\cal F}$ from samples, we cannot possibly draw the right conclusion with probability $1$ or detect the most minute deviations of $p$ from ${\cal C}$ or ${\cal F}$. So our guarantee should be probabilistic, and there should be some ``gap'' between the sets $\cal C$ and $\cal F$. In sum, the problem is the following:

\medskip \framebox{
\begin{minipage}{15.5cm}
{\em $(d_1,d_2)$-Identity Testing}: Given an explicit description of a distribution $q$ over $\Sigma$, sample access to a distribution $p$ over $\Sigma$, and bounds $\epsilon_1 \ge 0$, and $\epsilon_2, \delta >0$, distinguish with probability at least $1-\delta$ between $d_1(p,q)\le \epsilon_1$ and $d_2(p,q) \ge \epsilon_2$, whenever $p$ satisfies one of these two inequalities.
\end{minipage}}

\medskip \noindent A related problem is when we have sample access to both $p$ and $q$. For example, we might be interested in whether two populations of insects have distributions that are close or far. The resulting problem is the following:

\medskip \framebox{
\begin{minipage}{15.5cm}
{\em $(d_1,d_2)$-Equivalence (or Closeness) Testing}: Given sample access to distributions $p$ and $q$ over $\Sigma$, and bounds $\epsilon_1 \ge 0$, and $\epsilon_2, \delta >0$, distinguish with probability at least $1-\delta$ between $d_1(p,q)\le \epsilon_1$ and $d_2(p,q) \ge \epsilon_2$, whenever $p,q$ satisfy one of these two inequalities.
\end{minipage}}

\medskip The above questions are of course fundamental, and widely studied since the beginning of statistics.  However, most tests only detect certain types of deviations of $p$ from $q$, or are designed for distributions in parametric families. Moreover, most of the emphasis has been on the asymptotic sample regime. To address these challenges, there has been a surge of recent interest in information theory, property testing, and sublinear-time algorithms aiming at finite sample and $d_1$-close vs. $d_2$-far distinguishers, as in the formulations above; see e.g.\ \cite{BatuFFKRW01,BatuKR04,Paninski08,ValiantV17,AcharyaDK15,CanonneDGR16,DiakonikolasK16}. This line of work has culminated in computationally efficient and sample optimal testers for several choices of distances $d_1$ and $d_2$, as well as error parameters $\epsilon_1$ and $\epsilon_2$, for example:
\begin{itemize}
\item for identity testing, when: 
\begin{itemize}
\item $d_2$ is taken to be the total variation distance, and $\epsilon_1=0$~\cite{BatuFFKRW01, Paninski08,ValiantV17};

\item $d_1$ is taken to be the $\chi^2$-divergence, $d_2$ is taken to be the total variation distance, and $\epsilon_1=(\epsilon_2)^2/4$~\cite{AcharyaDK15, DiakonikolasK16};
\end{itemize}

\item for equivalence testing, when $d_2$ is taken to be the total variation distance, and $\epsilon_1=0$~\cite{BatuFRSW13, Valiant11, ChanDVV14}.

\end{itemize}
There are also several other sub-optimal results known for other combinations of $d_1$, $d_2$, $\epsilon_1$ and $\epsilon_2$, and for many combinations there are no known testers. A more extensive discussion of the literature is provided in Section~\ref{sec:related-work}.

\medskip The goal of this paper is to {\em provide a complete mapping of the optimal sample complexity required to obtain computationally efficient testers for identity testing and equivalence testing under the most commonly used notions of distances $d_1$ and $d_2$.} Our results are summarized in Tables~\ref{tab:id}, \ref{tab:eq}, and \ref{tab:l2} and discussed in detail in Section~\ref{sec:results}. In particular, we obtain computationally efficient and sample optimal testers for distances $d_1$ and $d_2$ ranging in the set \{$\ell_2$-distance, total variation distance, Hellinger distance, Kullback-Leibler divergence, $\chi^2$-divergence\},\footnote{These distances are nicely nested, as discussed in Section~\ref{sec:preliminaries}, from the weaker $\ell_2$ to the stronger $\chi^2$-divergence.} and for combinations of these distances and choice of errors $\epsilon_1$ and $\epsilon_2$ which give rise to meaningful testing problems as discussed above. The sample complexities stated in the tables are for probability of error $1/3$. Throwing in extra factors of $O(\log 1/\delta)$ boosts the probability of error to $1-\delta$, as usual.\footnote{Namely, one can repeat the test $O(\log 1/\d)$ times and output the majority result. One can analyze the resulting probability of success by the Chernoff bound.}

Our motivation for this work is primarily the fundamental nature of identity and equivalence testing, as well as of the distances under which we study these problems. It is also the fact that, even though distribution testing is by now a mature subfield of information theory, property testing, and sublinear-time algorithms, several of the testing questions that we consider have had unknown statuses prior to our work. This gap is accentuated by the fact that, as we establish, closely related distances may have radically different behavior. To give a quick example, it is easy to see that $\chi^2$-divergence is the second-order Taylor expansion of KL-divergence. Yet, as we show, the sample complexity for identity testing changes radically when $d_2$ is taken to be total variation or Hellinger distance, and $d_1$ transitions from $\chi^2$ to KL or weaker distances; see~Table~\ref{tab:id}. Prior to this work we knew about a transition somewhere between $\chi^2$-divergence and total variation distance, but our work identifies a more refined understanding of the point of transition. Similar fragility phenomena are identified by our work for equivalence testing, when we switch from total variation to Hellinger distance, as seen in Tables~\ref{tab:eq} and~\ref{tab:l2}. 

Adding to the fundamental nature of the problems we consider here, we should also emphasize that a clear understanding of the different tradeoffs mapped out by our work is critical at this point for the further development of the distribution testing field, as recent experience has established. Let us provide a couple of recent examples, drawing from our prior work. Acharya, Daskalakis, and Kamath~\cite{AcharyaDK15} study whether properties of distributions, such as unimodality or log-concavity, can be tested in total variation distance. Namely, given sample access to a distribution $p$, how many samples are needed to test whether it has some property (modeled by a set ${\cal P}$ of distributions) or whether it is far from having the property, i.e.~$d_{\rm TV}(p,{\cal P})>\epsilon$, for some  error $\epsilon$? Their approach is to first learn a proxy distribution $\hat{p} \in {\cal P}$ that satisfies $d'(p,\hat{p}) \le \epsilon'$ for some distance $d'$, whenever $p \in {\cal P}$, then reduce the property testing problem to $(d',d_{\rm TV})$-identity testing  of $p$ to $\hat{p}$. Interestingly, rather than picking $d'$ to be total variation distance, they take it to be $\chi^2$-divergence, which leads to optimal testers of sample complexity $O(\sqrt{n}/\epsilon^2)$ for several ${\cal P}$'s such as monotone, unimodal, and log-concave distributions over $[n]$. Had they picked $d'$ to be total variation distance, they would be stuck with a $\Omega(n/\log n)$ sample complexity in the resulting identity testing problem, as Table~\ref{tab:id} illustrates, which would lead to a suboptimal overall tester. The choice of $\chi^2$-divergence in the work of Acharya et al.~was somewhat ad hoc. By providing a full mapping of the sample complexity tradeoffs in the use of different distances, we expect to help future work in identifying better where the bottlenecks and opportunities lie.

Another example supporting our expectation can be found in recent work of Daskalakis and Pan~\cite{DaskalakisP17}. They study equivalence testing of Bayesian networks under total variation distance. Bayesian networks are flexible models expressing combinatorial structure in high-dimensional distributions in terms of a directed acyclic graph (DAG) specifying their conditional dependence structure. The challenge in testing Bayes nets is that their support scales exponentially in the number of nodes, and hence naive applications of known equivalence tests lead to sample complexities that are exponential in the number of nodes, even when the in-degree $\delta$ of the underlying DAGs is bounded. To address this challenge, Daskalakis and Pan establish ``localization-of-distance'' results of the following form, for various choices of distance $d$: ``If two Bayes nets $P$ and $Q$ are $\epsilon$-far in total variation distance, then there exists a small set of nodes $S$ (whose size is $\Delta+1$, where $\Delta$ is again the maximum in-degree of the underlying DAG where $P$ and $Q$ are defined) such that the marginal distributions of $P$ and $Q$ over the nodes of set $S$ are $\epsilon'$-far under distance $d$.'' When they take $d$ to be total variation distance, they can show $\epsilon'=\Omega(\epsilon/m)$, where $m$ is the number of nodes in the underlying DAG (i.e.~the dimension). Given this localization of distance, to test whether two Bayes nets $P$ and $Q$ satisfy $P=Q$ vs. $d_{\rm TV}(P,Q) \ge \epsilon$, it suffices to test, for all relevant marginals $P_S$ and $Q_S$ whether $P_S=Q_S$ vs. $d_{\rm TV}(P_S,Q_S)=\Omega(\epsilon/m)$. From Table~\ref{tab:eq} it follows that this requires sample size superlinear in $m$, which is suboptimal. Interestingly, when they take $d$ to be the square Hellinger distance, they can establish a localization-of-distance result with $\epsilon'=\epsilon^2/2m$. By Table~\ref{tab:eq}, to test each $S$ they need sample complexity that is linear in $m$, leading to an overall dependence of the sample complexity on $m$ that is $\tilde{O}(m)$,\footnote{The extra log factors are to guarantee that the tests performed on all sets $S$ of size $\delta+1$ succeed.} which is optimal up to log factors. Again, switching to a different distance results in near-optimal overall sample complexity, and our table is guidance as to where the bottlenecks and opportunities lie. 


Finally, we comment that tolerant testing (i.e., when $\ve_1 > 0$) is perhaps one of the most interesting questions in the design of practically useful testers.
Indeed, as mentioned before, in many statistical settings there may be model misspecification.
For example, why should one expect to be receiving samples from \emph{precisely} the uniform distribution?
As such, one may desire that a tester is \emph{robust} to small errors, and accepts all distributions which are \emph{close} to uniform.
Unfortunately, Valiant and Valiant~\cite{ValiantV11a} ruled out the possibility of a strongly sublinear tester which has total variation tolerance, showing that such a problem requires $\Theta\left(\frac{n}{\log n}\right)$ samples.
However, as shown by Acharya, Daskalakis, and Kamath~\cite{AcharyaDK15}, $\chi^2$-tolerance is possible with only $O\left(\frac{\sqrt{n}}{\ve^2}\right)$ samples.
This raises the following question: Which distances can a tester be tolerant to, while maintaining a strongly sublinear sample complexity? We outline what is possible.

\subsection{Results} \label{sec:results}
Our results are pictorially presented in Tables~\ref{tab:id}, \ref{tab:eq}, and \ref{tab:l2}.
We note that these tables are intended to provide only references to the \emph{sample complexity} of each testing problem, rather than exhaustively cover all prior work. 
As such, several references are deferred to Section~\ref{sec:related-work}.
In Tables~\ref{tab:id} and \ref{tab:eq}, each cell contains the complexity of testing whether two distributions are close in the distance for that row, versus far in the distance for that column.\footnote{Note that we chose constants in our theorem statements for simplicity of presentation, and they may not match the constants presented in the table. This can be remedied by appropriate changing of constants in the algorithms and constant factor increases in the sample complexity.}
These distances and their relationships are covered in detail in Section~\ref{sec:preliminaries}, but we note that the distances are scaled and transformed such that problems become harder as we traverse the table down or to the right.
In other words, lower bounds hold for cells which are down or to the right in the table, and upper bounds hold for cells which are up or to the left; problems with the same complexity are shaded with the same color.
The dark grey boxes indicate problems which are not well-defined, i.e. two distributions could simultaneously be close in KL and far in $\chi^2$-divergence.
\input{idtable}
\input{equivtable}

\input{l2table}

We highlight some of our results:
\begin{enumerate}
\item We give an $O(\sqrt{n}/\ve^2)$ sample algorithm for identity testing whether $\dxs(p,q) \leq \ve^2/4$ or $\dh(p,q) \geq \ve/\sqrt{2}$ (Theorem \ref{thm:ones-csq-h}).
This is the first algorithm which achieves the optimal dependence on both $n$ and $\ve$ for identity testing with respect to Hellinger distance (even non-tolerantly).
We note that a $O(\sqrt{n}/\ve^4)$ algorithm was known, due to optimal identity testers for total variation distance and the quadratic relationship between total variation and Hellinger distance.
\item In the case of identity testing, a stronger form of tolerance (i.e., KL divergence instead of $\chi^2$) causes the sample complexity to jump to $\Omega\left(n/\log n\right)$ (Theorem~\ref{thm:ones-lb}).
We find this a bit surprising, as $\chi^2$-divergence is the second-order Taylor expansion of KL divergence, so one might expect that the testing problems have comparable complexities.
\item In the case of equivalence testing, \emph{even $\chi^2$-tolerance} comes at the cost of an $\Omega\left(n/\log n\right)$ sample complexity (Theorem~\ref{thm:twos-lb}).
This is a qualitative difference from identity testing, where $\chi^2$-tolerance came at no cost.
\item However, in both identity and equivalence testing, $\ell_2$ tolerance comes at no additional cost (Theorems~\ref{thm:ones-tv}, \ref{thm:ones-h}, \ref{thm:twos-tv}, and \ref{thm:twos-h}).
Thus, in many cases, $\ell_2$ tolerance is the best one can do if one wishes to maintain a strongly sublinear sample complexity.
\end{enumerate}

From a technical standpoint, our algorithms are $\chi^2$-statistical tests, and most closely resemble those of~\cite{AcharyaDK15} and~\cite{ChanDVV14} (similar $\chi^2$-tests were employed in~\cite{ValiantV17, DiakonikolasKN15a, CanonneDGR16}).
However, crucial changes are required to satisfy the more stringent requirements of testing with respect to Hellinger distance.
In our identity tester for Hellinger, we deal with this different distance measure by pruning light domain elements of $q$ less aggressively than~\cite{AcharyaDK15},
in combination with a preliminary test to reject early if the difference between $p$ and $q$ is contained exclusively within the set of light elements -- this is a new issue that cannot arise when testing with respect to total variation distance.
In our equivalence tester for Hellinger, we follow an approach, similar to~\cite{ChanDVV14} and~\cite{DiakonikolasK16},
of analyzing the light and heavy domain elements separately, with the challenge that the algorithm does not know which elements are which.
Finally, to achieve $\ell_2$ tolerance in these cases, we use a ``mixing" strategy
in which instead of testing based solely on samples from~$p$ and~$q$, we mix in some number (depending on our application) of samples from the uniform distribution.
At a high level, the purpose of mixing is to make our distributions \emph{well-conditioned},
i.e.\ to ensure that all probability values are sufficiently large.
Such a strategy was recently employed by Goldreich in~\cite{Goldreich16} for uniformity testing.

\subsubsection{Comments on $\ell_2$-tolerance}
$\ell_2$ tolerance has been indirectly considered in~\cite{GoldreichR00, BatuFFKRW01, BatuFRSW13} through their weak tolerance for total variation distance and the relationship with $\ell_2$ distance, though these results have suboptimal sample complexity.
Our equivalence testing results improve upon~\cite{ChanDVV14} by adding $\ell_2$-tolerance.
We note that~\cite{DiakonikolasK16} also provides $\ell_2$-tolerant testers (as well as~\cite{DiakonikolasKN15a} for the case of uniformity), comparable to those obtained in Theorems~\ref{thm:ones-tv},~\ref{thm:ones-h}, and~\ref{thm:twos-h}, though this tolerance is not explicitly analyzed in their paper.
This can be seen by noting that the underlying tester from~\cite{ChanDVV14} is tolerant, and the ``flattening'' operation they apply reduces the $\ell_2$-distance between the distributions.
The testers in~\cite{DiakonikolasK16} are those of Propositions 2.7, 2.10, and 2.15, combined with the observation of Remark 2.8.
We rederive these results for completeness, and to show a direct way of proving $\ell_2$-tolerance.
Note that Theorem~\ref{thm:twos-h} also improves upon Proposition 2.15 of~\cite{DiakonikolasK16} by removing log factors in the sample complexity.

\subsubsection{Comments on the $\Theta(n/\log n)$ Results}
\label{sec:nlogn}
Our upper bounds in the bottom-left portion of the table are based off the total variation distance estimation algorithm of Jiao, Han, and Weissman~\cite{JiaoHW16}, where an $\Theta(n/\log n)$ complexity is only derived for $\ve \geq 1/\poly(n)$.
Similarly, in~\cite{ValiantV10a}, the lower bounds are only valid for constant $\ve$.
We believe that the precise characterization  is a very interesting open problem.
In the present work, we focus on the case of constant $\ve$ for these testing problems.

We wish to draw attention to the bottom row of the table, and note that the two testing problems are $\dtv(p,q) \leq \ve/2$ versus $\dtv(p,q) \geq \ve$, and $\dtv(p,q) \leq \ve^2/4$ versus $\dh(p,q) \geq \ve/\sqrt{2}$.
This difference in parameterization is required to make the two cases in the testing problem disjoint.
With this parameterization, we conjecture that the latter problem has a greater dependence on $\ve$ as it goes to $0$ (namely, $\ve^{-4}$ versus $\ve^{-2}$), so we colour the box a slightly darker shade of orange.

\subsection{Related Work}
\label{sec:related-work}
\input{related}

\subsection{Organization}
The organization of this paper is as follows.
In Section~\ref{sec:preliminaries}, we state preliminaries and notation used in this paper.
In Sections~\ref{sec:ones-ub} and~\ref{sec:twos-ub}, we prove upper bounds for identity testing and equivalence testing (respectively) based on $\chi^2$-style statistics.
In Section~\ref{sec:est-ub}, we prove upper bounds for distribution testing based on distance estimation.
Finally, in Section~\ref{sec:lb}, we prove testing lower bounds.

%% file: idtable.tex
  \newlength{\lena}
  \settowidth{\lena}{$\Omega\left(\frac{\sqrt{n}}{\ve^2}\right)$\cite{Paninski08}}
  \newlength{\lenb}
  \settowidth{\lenb}{$O\left(\frac{\sqrt{n}}{\ve^2}\right)$[Theorem~\ref{thm:ones-csq-h}]}
  \newlength{\lenc}
  \settowidth{\lenc}{Untestable [Theorem~\ref{thm:untestable}]}
  \newlength{\lend}
  \settowidth{\lend}{$\Omega\left(\frac{n}{\log n}\right)$[Theorem~\ref{thm:ones-lb}]}
  \newlength{\lene}
  \settowidth{\lene}{$O\left(\frac{n}{\log n}\right)$[Corollary~\ref{cor:tv-est}]}

\bgroup
\def\arraystretch{1.3}%
\begin{table}[t]
\centering
\scalebox{0.75}{
  \begin{tabular}{| c | c | c | c | c | c |}
  \hline
  {} & {$\dtv(p,q) \geq \ve$ } & {$\dh(p,q) \geq \ve/\sqrt{2}$} & {$\dkl(p,q) \geq \ve^2$} & {$\dxs(p,q) \geq \ve^2$} \\
  \hline
  {$p = q$} & \cellcolor{green!20}{\rlap{\hspace{0.5\lena}\tikzmark{l}}$\Omega\left(\frac{\sqrt{n}}{\ve^2}\right)$\cite{Paninski08}\tikzmark{a}} & \cellcolor{green!20}{\tikzmark{b}} & \cellcolor{red!20}{\rlap{\hspace{0.5\lenc}\tikzmark{o}}Untestable [Theorem~\ref{thm:untestable}]\tikzmark{c}} & \cellcolor{red!20}{}  \\
  \hline
  {$\dxs(p,q) \leq \ve^2/4$} & \cellcolor{green!20}{\tikzmark{e}} & \cellcolor{green!20}{\tikzmark{f}\rlap{\hspace{0.5\lenb}\tikzmark{m}}$O\left(\frac{\sqrt{n}}{\ve^2}\right)$[Theorem~\ref{thm:ones-csq-h}]} & \cellcolor{red!20}{} & \cellcolor{red!20}{}  \\
  \hline
  {$\dkl(p,q) \leq \ve^2/4$} & \cellcolor{orange!20}{\rlap{\hspace{0.5\lend}\tikzmark{p}}$\Omega\left(\frac{n}{\log n}\right)$[Theorem~\ref{thm:ones-lb}]\tikzmark{g}} & \cellcolor{orange!20}{\tikzmark{h}} & \cellcolor{red!20}{} & \cellcolor{black!20}{} \\
  \hline
  {$\dh(p,q) \leq \ve/2\sqrt{2}$} & \cellcolor{orange!20}{} & \cellcolor{orange!20}{} & \cellcolor{black!20}{} & \cellcolor{black!20}{}  \\
  \hline
  {$\dtv(p,q) \leq \ve/2$ or $\ve^2/4$\footnotemark} & \cellcolor{orange!20}{\tikzmark{i}} & \cellcolor{orange!25}{\tikzmark{j}\rlap{\hspace{0.5\lene}\tikzmark{q}}$O\left(\frac{n}{\log n}\right)$[Corollary~\ref{cor:tv-est}]} & \cellcolor{black!20}{\tikzmark{n}} & \cellcolor{black!20}{}  \\
  \hline
\end{tabular}
}
\caption{Identity Testing. Rows correspond to completeness of the tester, and columns correspond to soundness.}
\label{tab:id}
\end{table}

\egroup
\footnotetext{\label{fn:nlogn}We note that we must use $\ve/2$ or $\ve^2/4$ depending on whether we are testing with respect to TV or Hellinger. For more details and other discussion of the $n/\log n$ region of this chart, see Section~\ref{sec:nlogn}.}

%% file: equivtable.tex
  \newlength{\lenaa}
  \settowidth{\lenaa}{$\Omega\left(\frac{n}{\log n}\right)$ [Theorem~\ref{thm:twos-lb}]}
  \newlength{\lenbb}
  \settowidth{\lenbb}{$O\left(\frac{n}{\log n}\right)$[Corollary~\ref{cor:tv-est}]}

\bgroup
\def\arraystretch{1.3}%

\begin{table}[t]
\centering
\scalebox{0.75}{
  \begin{tabular}{| c | c | c | c | c | c |}
  \hline
  {} & {$\dtv(p,q) \geq \ve$ } & {$\dh(p,q) \geq \ve/\sqrt{2}$} & {$\dkl(p,q) \geq \ve^2$} & {$\dxs(p,q) \geq \ve^2$}  \\
  \hline
  {$p = q$} & \cellcolor{violet!20}{$O \left(\max\left\{\frac{n^{1/2}}{\ve^2},\frac{n^{2/3}}{\ve^{4/3}}\right\}\right)$\cite{ChanDVV14}} & \cellcolor{yellow!20}{$O\left(\min\left\{\frac{n^{3/4}}{\ve^2},\frac{n^{2/3}}{\ve^{8/3}}\right\}\right)$[Theorem~\ref{thm:twos-h}]} & \cellcolor{red!20}{\rlap{\hspace{0.5\lenc}\tikzmark{ii}}Untestable [Theorem~\ref{thm:untestable}]\tikzmark{aa}} & \cellcolor{red!20}{}  \\
  {} & \cellcolor{violet!20}{$\Omega \left(\max\left\{\frac{n^{1/2}}{\ve^2},\frac{n^{2/3}}{\ve^{4/3}}\right\}\right)$\cite{ChanDVV14}} & \cellcolor{yellow!20}{$\Omega\left(\min\left\{\frac{n^{3/4}}{\ve^2},\frac{n^{2/3}}{\ve^{8/3}}\right\}\right)$\cite{DiakonikolasK16}} & \cellcolor{red!20}{} & \cellcolor{red!20}{}  \\
  \hline
  {$\dxs(p,q) \leq \ve^2/4$} & \cellcolor{orange!20}{\rlap{\hspace{0.5\lenaa}{\tikzmark{jj}}}$\Omega\left(\frac{n}{\log n}\right)$ [Theorem~\ref{thm:twos-lb}]\tikzmark{dd}} & \cellcolor{orange!20}{\tikzmark{ee}} & \cellcolor{red!20}{} & \cellcolor{red!20}{}  \\
  \hline
  {$\dkl(p,q) \leq \ve^2/4$} & \cellcolor{orange!20}{} & \cellcolor{orange!20}{} & \cellcolor{red!20}{} & \cellcolor{black!20}{} \\
  \hline
  {$\dh(p,q) \leq \ve/2\sqrt{2}$} & \cellcolor{orange!20}{} &\cellcolor{orange!20} {} & \cellcolor{black!20}{} & \cellcolor{black!20}{}  \\
  \hline
  {$\dtv(p,q) \leq \ve/2$ or $\ve^2/4$\textsuperscript{\ref{fn:nlogn}}} & \cellcolor{orange!20}{\tikzmark{ff}} & \cellcolor{orange!25}{\tikzmark{gg}\rlap{\hspace{0.5\lenbb}{\tikzmark{kk}}}$O\left(\frac{n}{\log n}\right)$[Corollary~\ref{cor:tv-est}]} & \cellcolor{black!20}{\tikzmark{hh}} & \cellcolor{black!20}{}  \\
  \hline
\end{tabular}}
\caption{Equivalence Testing. Rows correspond to completeness of the tester, and columns correspond to soundness.}
\label{tab:eq}
\end{table}
\egroup

%% file: l2table.tex
\bgroup
\def\arraystretch{1.3}%
\begin{table}[t]
\centering
  \begin{tabular}{| c | c | c |}
  \hline
  {} & {Identity Testing} & {Equivalence Testing}\\
  \hline
  {$d(p,q) \leq f_d(n,\ve)$ vs. $d_{\ell_2}(p,q) \geq \ve$} & \cellcolor{blue!20}{$\Theta\left(\frac1{\ve^2}\right)$ [Corollary \ref{cor:l2-est}]} & \cellcolor{blue!20}{$\Theta\left(\frac1{\ve^2}\right)$ [Corollary \ref{cor:l2-est}]} \\
  \hline
  {$d_{\ell_2}(p,q) \leq \frac{\ve}{\sqrt{n}}$ vs. $\dtv(p,q) \geq \ve$} & \cellcolor{green!20}{$\Theta\left(\frac{\sqrt{n}}{\ve^2}\right)$ [Theorem~\ref{thm:ones-tv}]} & \cellcolor{violet!20}{$\Theta\left(\max\left\{\frac{n^{1/2}}{\ve^2},\frac{n^{2/3}}{\ve^{4/3}}\right\}\right)$ [Theorem~\ref{thm:twos-tv}]} \\
  \hline
  {$d_{\ell_2}(p,q) \leq \frac{\ve^2}{\sqrt{n}}$ vs. $\dh(p,q) \geq \ve$} & \cellcolor{green!20}{$\Theta\left(\frac{\sqrt{n}}{\ve^2}\right)$ [Theorem~\ref{thm:ones-h}]} & \cellcolor{yellow!20}{$\Theta\left(\min\left\{\frac{n^{3/4}}{\ve^2},\frac{n^{2/3}}{\ve^{8/3}}\right\}\right)$[Theorem~\ref{thm:twos-h}}] \\
  \hline
\end{tabular}
\caption{$\ell_2$ Testing. $f_d(n, \ve)$ is a quantity such that $d(p,q) \leq f_d(n, \ve)$ and $d_{\ell_2}(p,q) \geq \ve$ are disjoint.}
\label{tab:l2}
\end{table}

\egroup

%% file: related.tex
The most classic distribution testing question is uniformity testing, which is identity testing when $\ve_1 = 0$, $d_2$ is total variation distance, and $q$ is the uniform distribution.
This was first studied in theoretical computer science in~\cite{GoldreichR00}.
Paninski gave an optimal algorithm (for when $\ve_2$ is not too small) with a complexity of $O(\sqrt{n}/\ve^2)$ and a matching lower bound~\cite{Paninski08}.
More generally, letting $q$ be an arbitrary distribution, exact total variation identity testing was studied~\cite{BatuFFKRW01}, and an (instance) optimal algorithm was given by Valiant and Valiant~\cite{ValiantV17}, with the same complexity of $O(\sqrt{n}/\ve^2)$.
Optimal algorithms for this problem were rediscovered several times, see i.e.~\cite{DiakonikolasKN15a, AcharyaDK15, DiakonikolasK16, DiakonikolasGPP16}.

Equivalence (or closeness) testing was studied in~\cite{BatuFRSW13}, in the same setting ($\ve_1 = 0$, $d_2$ is total variation distance).
A lower bound of $\Omega(n^{2/3})$ was given by~\cite{Valiant11}.
Tight upper and lower bounds were given in~\cite{ChanDVV14}, which shows interesting behavior of the sample complexity as the parameter $\ve$ goes from large to small.
This problem was also studied in the setting where one has unequal sample sizes from the two distributions~\cite{BhattacharyaV15,DiakonikolasK16}.
When the distance $d_1$ is Hellinger, the complexity is qualitatively different, as shown by~\cite{DiakonikolasK16}.
They prove a nearly-optimal upper bound and a tight lower bound for this problem.

\cite{Waggoner15, DoBaNNR11} also consider testing problems with other distances, namely $\ell_p$ distances and earth mover's distance (also known as Wasserstein distance), respectively.

Tolerant identity testing (where $\ve_1 = O(\ve)$ and $d_1$ is total variation distance) was studied in~\cite{ValiantV10a,ValiantV10b,ValiantV11a,ValiantV11b}, through the (equivalent) lens of estimating total variation distance between distributions.
In these works, $\Theta\left(n/\log n\right)$ bounds were proven for the sample complexity.
Several other related problems (i.e., support size and entropy estimation) share the same sample complexity, and have enjoyed significant study~\cite{AcharyaOST17, WuY16, AcharyaDOS17}.
The closest related results to our work are those on estimating distances between distributions~\cite{JiaoHW16, JiaoKHW17, HanJW16}.

$\chi^2$-tolerance (when $d_1$ is $\chi^2$-divergence and $\ve_1 = O(\ve^2)$) was introduced and applied by~\cite{AcharyaDK15} for testing families of distributions, i.e., testing if a distribution is monotone or far from being monotone.
It was shown that this tolerance comes at no additional cost over vanilla identity testing; that is, the sample complexity is still $O(\sqrt{n}/\ve^2)$.
Testing such families of distributions was also studied by~\cite{CanonneDGR16}.

Testing with respect to Hellinger distance was applied in~\cite{DaskalakisP17} for testing Bayes networks.
Since lower bounds of~\cite{AcharyaDK15} show that distribution testing suffers from the curse of dimensionality, further structural assumptions must be made if one wishes to test multivariate distributions.
This ``high-dimensional frontier'' has also been studied on graphical models by~\cite{DaskalakisDK18} and~\cite{CanonneDKS17} (for Ising models and Bayesian networks, respectively).

Our work focuses on characterizing the complexity of identity and equivalence testing in the worst case over pairs $p$ and $q$.
Related works attempt to nail down the sample complexity of identity testing on an \emph{instance-by-instance} basis~\cite{ValiantV17, JiaoHW16, DiakonikolasK16, BlaisCG17} -- that is, reducing the sample complexity depending on which distribution $q$ is given as input (and sometimes depending on $p$ as well).
We consider this to be an interesting open question for different distances $d_1$ and $d_2$.
For example, Theorem~\ref{thm:untestable} states that identity testing is impossible when $d_2$ is the KL divergence.
However, if $q$ is the uniform distribution, then the complexity becomes $\Theta(\sqrt{n})$.
An instance-by-instance analysis would allow one to bypass some of these strong lower bounds.

This is only a fraction of recent results; we direct the reader to~\cite{Canonne15} for an excellent recent survey of distribution testing.

%% file: preliminaries.tex
\section{Preliminaries}
\label{sec:preliminaries}
In this paper, we will focus on discrete probability distributions over $[n]$.
For a distribution $p$, we will use the notation $p_i$ to denote the mass $p$ places on symbol $i$.
For a set $S \subseteq [n]$ and a distribution $p$ over $[n]$, $p_S$ is the vector $p$ restricted to the coordinates in $S$.
We will call this a \emph{restriction} of distribution $p$.

The following probability distances and divergences are of interest to us:
\begin{definition}
The \emph{total variation distance} between $p$ and $q$ is defined as
$$\dtv(p,q) = \max_{S \subseteq [n]} p(S) - q(S) = \frac12\sum_{i \in [n]} \left|p_i - q_i\right| = \|p - q\|_1 \in [0,1].$$
\end{definition}

\begin{definition}
The \emph{KL divergence} between $p$ and $q$ is defined as
$$\dkl(p,q) = \sum_{i \in [n]} p_i \log \left(\frac{p_i}{q_i}\right) \in [0, \infty).$$
This definition uses the convention that $0 \log 0 = 0$.
\end{definition}

\begin{definition}
The \emph{Hellinger distance} between $p$ and $q$ is defined as 
$$\dh(p,q) = \frac{1}{\sqrt{2}}\sqrt{\sum_{i \in [n]} \left(\sqrt{p_i} - \sqrt{q_i}\right)^2} \in [0,1].$$
\end{definition}

\begin{definition}
The \emph{$\chi^2$-divergence} (or \emph{chi-squared divergence}) between $p$ and $q$ is defined as
$$\dxs(p,q) = \sum_{i \in [n]} \frac{(p_i - q_i)^2}{q_i} \in [0,\infty).$$
\end{definition}

\begin{definition}
The \emph{$\ell_2$ distance} between $p$ and $q$ is defined as
$$\dlt(p,q) = \sqrt{\sum_{i \in [n]} (p_i - q_i)^2} = \|p-q\|_2 \in [0,1].$$
\end{definition}
\noindent
We also define these distances for restrictions of distributions $p_S$ and $q_S$ by replacing the summations over $i \in [n]$ with summations over $i \in S$.

We have the following relationships between these distances.
These are well-known for distributions, i.e., see \cite{GibbsS02}, but we prove them more generally for restrictions of distributions in Section~\ref{sec:distanceinequalities}. 
\begin{proposition}
\label{prop:distanceinequalities}
Letting $p_S$ and $q_S$ be restrictions of distributions $p$ and $q$ to $S \subseteq [n]$,
$$\dh^2(p_S,q_S) \leq \dtv(p_S,q_S) \leq \sqrt{2}\dh(p_S,q_S) \leq \sqrt{\sum_{i \in S} (q_i - p_i)  + \dkl(p_S,q_S)} \leq \sqrt{\dxs(p_S,q_S)}.$$
\end{proposition}

We recall that $\dlt$ fits into the picture by its relationship with total variation distance:
\begin{proposition}
\label{prop:ltinequalities}
Letting $p$ and $q$ be distributions over $[n]$,
$$\dlt(p,q) \leq 2\dtv(p,q) \leq \sqrt{n}\dlt(p,q).$$
\end{proposition}
The second inequality follows from Cauchy-Schwarz.

We will also need to following bound for Hellinger distance:
\begin{proposition}\label{prp:didnt-know-it-was-hellinger}
$\displaystyle 2 \dh^2(p,q) \leq \sum_{i=1}^n \frac{(p_i - q_i)^2}{p_i+q_i} \leq 4\dh^2(p,q)$.
\end{proposition}
\begin{proof}
Expanding the Hellinger-squared distance,
\begin{equation*}
\dh^2(p,q) = \frac12 \sum_{i=1}^n (\sqrt{p_i} - \sqrt{q_i})^2 = \frac12 \sum_{i=1}^n \frac{(p_i-q_i)^2}{(\sqrt{p_i}+\sqrt{q_i})^2}.
\end{equation*}
The fact now follows because $(p_i + q_i) \leq (\sqrt{p_i}+\sqrt{q_i})^2 \leq 2(p_i + q_i)$.
\end{proof}
\noindent
The quantity $\sum_{i=1}^n (p_i-q_i)^2/(p_i+q_i)$ is sometimes called the \emph{triangle distance}.
However, we see here that it is essentially the Hellinger distance (up to constant factors).

\begin{proposition}\label{prop:mixing}
Given a number $\delta \in [0,1]$ and a discrete distribution~$r  = (r_1, \ldots, r_n)$, define
\begin{equation*}
r^{+\delta} := (1-\delta) \cdot r + \delta \cdot (\tfrac{1}{n} , \ldots, \tfrac{1}{n}).
\end{equation*}
Then given two discrete distributions $p=(p_1, \ldots, p_n)$ and $q=(q_1, \ldots, q_n)$,
\begin{equation*}
\dtv(p^{+\delta}, q^{+\delta}) = (1-\delta) \dtv(p,q),
\quad
\dlt(p^{+\delta},q^{+\delta})= (1-\delta) \dlt(p,q).
\end{equation*}
In addition, $\dh(p^{+\delta},q^{+\delta}) \geq \dh(p,q) - 2\sqrt{\delta}$.
\end{proposition}
\begin{proof}
The statements for total variation and~$\ell_2$ distance are immediate.  As for the Hellinger distance, we have by the triangle inequality that
\begin{equation*}
\dh(p,q) \leq \dh(p,p^{+\delta}) + \dh(p^{+\delta},q^{+\delta}) + \dh(q^{+\delta},q).
\end{equation*}
We can bound the first term by
\begin{equation*}
\dh^2(p,p^{+\delta})
\leq \dtv( p, p^{+\delta})
= \tfrac{1}{2} \cdot\Vert \delta \cdot p - \delta \cdot (\tfrac{1}{n}, \ldots, \tfrac{1}{n}) \Vert_1
\leq \delta,
\end{equation*}
where the last step is by the triangle inequality,
and a similar argument bounds the third term by~$\sqrt{\delta}$ as well.
Thus, $\dh(p^{+\delta},q^{+\delta}) \geq \dh(p,q) - 2\sqrt{\delta}$.
\end{proof}
A similar technique was employed in~\cite{Goldreich16}.

At times, our algorithms will employ \emph{Poisson sampling}.
Instead of taking $m$ samples from a distribution $p$, we instead take $\mathrm{Poisson}(m)$ samples.
As a result, letting $N_i$ be the number of occurences of symbol $i$, all $N_i$ will be independent and distributed as $\mathrm{Poisson}(m \cdot p_i)$.
We note that this method of sampling is for purposes of analysis -- concentration bounds imply that $\mathrm{Poi}(m) = O(m)$ with high probability, so such an algorithm can be converted to one with a fixed budget of samples at a constant-factor increase in the sample complexity.

%% file: ones-ub.tex
\section{Upper Bounds for Identity Testing}
\label{sec:ones-ub}
In this section, we prove the following theorems for identity testing.
\begin{theorem}\label{thm:ones-csq-h}
There exists an algorithm for identity testing between $p$ and $q$ distinguishing the cases:
\begin{itemize}
\item $\dxs(p,q) \leq \ve^2$;
\item $\dh(p,q) \geq \ve$.
\end{itemize}
The algorithm uses $O\left(\frac{n^{1/2}}{\ve^2}\right)$ samples.
\end{theorem}

\begin{theorem}\label{thm:ones-tv}
There exists an algorithm for identity testing between $p$ and $q$ distinguishing the cases:
\begin{itemize}
\item $\dlt(p,q) \leq \frac{\ve}{\sqrt{n}}$;
\item $\dtv(p,q) \geq \ve$.
\end{itemize}
The algorithm uses $O\left(\frac{n^{1/2}}{\ve^2}\right)$ samples.
\end{theorem}

\begin{theorem}\label{thm:ones-h}
There exists an algorithm for identity testing between $p$ and $q$ distinguishing the cases:
\begin{itemize}
\item $\dlt(p,q) \leq \frac{\ve^2}{\sqrt{n}}$;
\item $\dh(p,q) \geq \ve$.
\end{itemize}
The algorithm uses $O\left(\frac{n^{1/2}}{\ve^2}\right)$ samples.
\end{theorem}

We prove Theorem~\ref{thm:ones-csq-h} in Section~\ref{sec:id-csq-h}, and Theorems~\ref{thm:ones-tv} and~\ref{thm:ones-h} in Section~\ref{sec:id-lt}.

\subsection{Identity Testing with Hellinger Distance and $\chi^2$-Tolerance}
\label{sec:id-csq-h}

We prove Theorem~\ref{thm:ones-csq-h} by analyzing Algorithm~\ref{alg:testing}.
We will set $c_1 = \frac{1}{100}, c_2 = \frac{6}{25}$, and let $C$ be a sufficiently large constant.
\begin{algorithm}[h]
\caption{$\chi^2$-close versus Hellinger-far testing algorithm}\label{alg:testing}
\begin{algorithmic}[1]
\State \textbf{Input:} $\ve$; an explicit distribution $q$; sample access to a distribution $p$
\State Implicitly define $\mathcal{A} \leftarrow \{i:q_i \geq c_1\ve^2/n\}$, $\mathcal{\bar A} \leftarrow [n] \setminus \mathcal{A}$
\State Let $\hat p$ be the empirical distribution\footnote{The empirical distribution is defined by taking a set of samples and normalizing the counts such that the result forms a probability distribution.} from drawing $m_1 = \Theta(1/\ve^2)$ samples from $p$
\If {$\hat p(\mathcal{\bar A}) \geq \frac34 c_2\ve^2$} \label{ln:light-test}
\State \Return \reject \label{ln:early-reject}
\EndIf
\State Draw a multiset $S$ of $\mathrm{Poisson}(m_2)$ samples from $p$, where $m_2 = C\sqrt{n}/\ve^2$
\State Let $N_i$ be the number of occurrences of the $i$th domain element in $S$
\State Let $S'$ be the set of domain elements observed in $S$
\State $Z \leftarrow \sum_{i \in S' \cap \mathcal{A}} \frac{(N_i - m_2q_i)^2 - N_i}{m_2q_i} + m_2 (1 - q(S' \cap \mathcal{A}))$ \label{ln:statistic}
\If {$Z \leq \frac{3}{2}m_2\ve^2$}
\State \Return \accept
\Else 
\State \Return \reject
\EndIf 
\end{algorithmic}
\end{algorithm}

We note that the sample and time complexity are both $O(\sqrt{n}/\ve^2)$.
We draw $m_1 + m_2 = \Theta(\sqrt{n}/\ve^2)$ samples total.
All steps of the algorithm only involve inspecting domain elements where a sample falls, and it runs linearly in the number of such elements.
Indeed, Step~\ref{ln:statistic} of the algorithm is written in an unusual way in order to ensure the running time of the algorithm is linear.

We first analyze the test in Step \ref{ln:light-test} of the algorithm.
Folklore results state that with probability at least $99/100$, this preliminary test will reject any $p$ with $p(\mathcal{\bar A}) \geq c_2 \ve^2$, it will not reject any $p$ with $p(\mathcal{\bar A}) \leq \frac{c_2}{2} \ve^2$, and behavior for any other $p$ is arbitrary.
Condition on the event the test does not reject for the remainder of the proof.
Note that since both thresholds here are $\Theta(\ve^2)$, it only requires $m_1 = \Theta(1/\ve^2)$ samples, rather than the ``non-extreme'' regime, where we would require $\Theta(1/\ve^4)$ samples.

\begin{remark}
We informally refer to this ``extreme'' versus ``non-extreme'' regime in distribution testing.
To give an example of what we mean in these two cases, consider distinguishing $Ber(1/2)$ from $Ber(1/2 + \ve)$.
The complexity of this problem is $\Theta(1/\ve^2)$, and we consider this to be in the non-extreme regime.
On the other hand, distinguishing $Ber(\ve)$ from $Ber(2\ve)$ has a sample complexity of $\Theta(1/\ve)$, and we consider this to be in the extreme regime.
\end{remark}

We justify that any $p$ which may be rejected in Step \ref{ln:early-reject} (i.e., any $p$ such that $p(\mathcal{\bar A}) > \frac{c_2}{2} \ve^2$) has the property that $\dxs(p,q) > \ve^2$ (in other words, we do not wrongfully reject any $p$).

Consider a $p$ such that $p(\mathcal{\bar A}) \geq \frac{c_2}{2}\ve^2$.
Note that $\dxs(p, q) \geq \dxs(p_\mathcal{\bar A}, q_\mathcal{\bar A})$, which we lower bound as follows:
\begin{align*}
\dxs(p_\mathcal{\bar A}, q_\mathcal{\bar A})
&= \sum_{i \in \mathcal{\bar A}} \frac{(p_i - q_i)^2}{q_i} \\
&\geq \frac{n}{c_1 \ve^2} \sum_{i \in \mathcal{\bar A}} (p_i - q_i)^2 \\
&\geq \frac{n}{c_1 \ve^2} \cdot \frac{1}{n} \left(\sum_{i \in \mathcal{\bar A}} (p_i - q_i) \right)^2  \\
&\geq \frac{n}{c_1 \ve^2} \frac{\ve^4\left(\frac{c_2}{2} - c_1\right)^2}{n} \\
&= \frac{\left(\frac{c_2}{2} - c_1\right)^2}{c_1}\ve^2
\end{align*}
The first inequality is by the definition of $\mathcal{\bar A}$, the second is by Cauchy-Schwarz, and the third is since $p(\mathcal{\bar A}) \geq \frac{c_2}{2}\ve^2$ and $q(\mathcal{\bar A}) \leq c_1\ve^2$.
By our setting of $c_1$ and $c_2$, this implies that $\dxs(p, q) > \ve^2$, and we are not rejecting any $p$ which should be accepted.

For the remainder of the proof, we will implicitly assume that $p(\mathcal{\bar A}) \leq c_2 \ve^2$.

Let
$$Z' = \sum_{i \in \mathcal{A}} \frac{(N_i - m_2 q_i)^2 - N_i}{m_2q_i}.$$

Note that the statistic $Z$ can be rewritten as follows:
\begin{align*}
Z &= \sum_{i \in S' \cap \mathcal{A}} \frac{(N_i - m_2q_i)^2 - N_i}{m_2q_i} + m_2 (1 - q(S' \cap \mathcal{A})) \\
  &= \sum_{i \in S' \cap \mathcal{A}} \frac{(N_i - m_2q_i)^2 - N_i}{m_2q_i} + \sum_{i \in \mathcal{A} \setminus S'} m_2 q_i + m_2  q(\mathcal{\bar A}) \\
  &= \sum_{i \in S' \cap \mathcal{A}} \frac{(N_i - m_2q_i)^2 - N_i}{m_2q_i} + \sum_{i \in \mathcal{A} \setminus S'} \frac{(N_i - m_2q_i)^2 - N_i}{m_2q_i} + m_2  q(\mathcal{\bar A}) \\
  &= Z' + m_2 q(\mathcal{\bar A})
\end{align*}

We proceed by analyzing $Z'$.
First, note that it has the following expectation and variance:
\begin{align}
\E[Z'] &= m_2 \cdot \sum_{i \in \mathcal{A}} \frac{(p_i - q_i)^2}{q_i} = m_2 \cdot \dxs(p_\mathcal{A}, q_\mathcal{A}) \label{eqn:mean} \\
\Var[Z'] &= \sum_{i \in \mathcal{A}} \left[2\frac{p_i^2}{q_i^2} + 4m_2 \cdot \frac{p_i \cdot (p_i - q_i)^2}{q_i^2}\right] \label{eqn:variance}
\end{align}
These properties are proven in Section A of~\cite{AcharyaDK15}.

We require the following two lemmas, which state that the mean of the statistic is separated in the two cases, and that the variance is bounded.
The proofs largely follow the proofs of two similar lemmas in~\cite{AcharyaDK15}.
\begin{lemma}
\label{lem:means}
If $\dxs(p,q) \leq \ve^2$, then $\E[Z'] \leq m_2 \ve^2$. 
If $\dh(p,q) \geq \ve$, then $\E[Z'] \geq (2 - c_1 - c_2)m_2 \ve^2$.
\end{lemma}
\begin{proof}
The former case is immediate from (\ref{eqn:mean}).

For the latter case, note that
$$\dh^2(p,q) = \dh^2(p_\mathcal{A}, q_\mathcal{A}) + \dh^2(p_\mathcal{\bar A}, q_\mathcal{\bar A}).$$
We upper bound the latter term as follows:
\begin{align*}
\dh^2(p_\mathcal{\bar A}, q_\mathcal{\bar A}) 
&\leq \dtv(p_\mathcal{\bar A}, q_\mathcal{\bar A}) \\
&= \frac12 \sum_{i \in \mathcal{\bar A}} |p_i - q_i| \\
&\leq \frac12 \left(p(\mathcal{\bar A}) + q(\mathcal{\bar A})\right) \\
&\leq \left(\frac{c_1 + c_2}{2}\right)\ve^2 \\
\end{align*}
The first inequality is from Proposition \ref{prop:distanceinequalities}, and the third inequality is from our prior condition that $p(\mathcal{\bar A}) \leq c_2 \ve^2$.

Since $\dh^2(p,q) \geq \ve^2$, this implies $\dh^2(p_\mathcal{A}, q_\mathcal{A}) \geq \left(1 - \frac{c_1 + c_2}{2}\right)\ve^2$.
Proposition \ref{prop:distanceinequalities} further implies that $\dxs(p_\mathcal{A}, q_\mathcal{A}) \geq \left(2 - c_1 - c_2\right)\ve^2$.
The lemma follows from (\ref{eqn:mean}).
\end{proof}

\begin{lemma}
\label{lem:vars}
If $\dxs(p,q) \leq \ve^2$, then $\Var[Z'] = O(m_2^2 \ve^4)$. 
If $\dh(p,q) \geq \ve$, then $\Var[Z'] \leq O(\E[Z']^2)$.
The constant in both expressions can be made arbitrarily small with the choice of the constant $C$.
\end{lemma}
\begin{proof}
We bound the terms of (\ref{eqn:variance}) separately, starting with the first.

\begin{align}
2\sum_{i \in \mathcal{A}} \frac{p_i^2}{q_i^2} &= 2\sum_{i \in \mathcal{A}} \left(\frac{(p_i - q_i)^2}{q_i^2} + \frac{2p_iq_i - q_i^2}{q_i^2}\right) \nonumber \\
                                             &= 2\sum_{i \in \mathcal{A}} \left(\frac{(p_i - q_i)^2}{q_i^2} + \frac{2q_i(p_i - q_i) + q_i^2}{q_i^2}\right) \nonumber\\
                                             &\leq 2n + 2\sum_{i \in \mathcal{A}} \left(\frac{(p_i - q_i)^2}{q_i^2} + 2\frac{(p_i - q_i)}{q_i}\right) \nonumber\\
                                             &\leq 4n + 4\sum_{i \in \mathcal{A}} \frac{(p_i - q_i)^2}{q_i^2} \nonumber\\
                                             &\leq 4n + \frac{4n}{c_1\ve^2} \sum_{i \in \mathcal{A}} \frac{(p_i - q_i)^2}{q_i}\nonumber\\
                                             &= 4n + \frac{4n}{c_1\ve^2}\frac{E[Z']}{m_2} \nonumber\\
                                             &\leq 4n + \frac{4}{c_1C}\sqrt{n} E[Z']\label{eq:first-var-term-in}
\end{align}
The second inequality is the AM-GM inequality, the third inequality uses that $q_i \geq \frac{c_1\ve^2}{n}$ for all $i \in \mathcal{A}$, the last equality uses \eqref{eqn:mean}, and the final inequality substitutes a value $m_2 \geq C\frac{\sqrt{n}}{\ve^2}$.

The second term can be similarly bounded:
\begin{align*}
4m_2 \sum_{i \in \mathcal{A}} \frac{p_i(p_i - q_i)^2}{q_i^2} &\leq 4m_2 \left(\sum_{i \in \mathcal{A}} \frac{p_i^2}{q_i^2}\right)^{1/2}\left(\sum_{i \in \mathcal{A}} \frac{(p_i - q_i)^4}{q_i^2}\right)^{1/2} \\
                                                          &\leq 4m_2 \left(4n + \frac{4}{c_1C}\sqrt{n} E[Z'] \right)^{1/2}\left(\sum_{i \in \mathcal{A}} \frac{(p_i - q_i)^4}{q_i^2}\right)^{1/2} \\
                                                          &\leq 4m_2 \left(2\sqrt{n} + \frac{2}{\sqrt{c_1C}}n^{1/4} E[Z']^{1/2}\right)\left(\sum_{i \in \mathcal{A}} \frac{(p_i - q_i)^2}{q_i}\right) \\
                                                          &= \left(8\sqrt{n} + \frac{8}{\sqrt{c_1C}}n^{1/4} E[Z']^{1/2}\right)E[Z'].
\end{align*}
The first inequality is Cauchy-Schwarz, the second inequality uses (\ref{eq:first-var-term-in}), the third inequality uses the monotonicity of the $\ell_p$ norms, and the equality uses~\eqref{eqn:mean}.

Combining the two terms, we get
$$\Var[Z'] \leq 4n + \left(8 + \frac{4}{c_1C}\right)\sqrt{n} \E[Z']  + \frac{8}{\sqrt{c_1C}}n^{1/4} \E[Z']^{3/2}  .$$

We now consider the two cases in the statement of our lemma.
\begin{itemize}
\item
When $\dxs(p,q) \leq \ve^2$, we know from Lemma~\ref{lem:means} that $\E[Z'] \leq m_2 \ve^2$. 
Combined with a choice of $m_2 \geq C \frac{\sqrt{n}}{\ve^2}$ and the above expression for the variance, this gives:
\begin{align*}
\Var[Z']
& \leq \frac{4}{C^2}m_2^2\ve^4 + \left(\frac{8}{C} + \frac{4}{c_1C^2}\right)m_2^2 \ve^4+ \frac{8}{C\sqrt{c_1}}m_2^2 \ve^4 \\
& = \left(\frac{8}{C} + \frac{8}{C\sqrt{c_1}} + \frac{4}{C^2} + \frac{4}{c_1C^2} \right)m_2^2 \ve^4 = O(m_2^2 \ve^4).
\end{align*}

\item When $\dh(p,q) \geq \ve$, Lemma~\ref{lem:means} and  $m_2 \geq C\frac{\sqrt{n}}{\ve^2}$ give:
$$\E[Z'] \geq (2 - c_1 - c_2) m_2 \ve^2 \geq C(2 - c_1 - c_2) \sqrt{n}.$$

Similar to before, combining this with our expression for the variance we get:
\begin{align*}
\Var[Z']
&\leq \left(\frac{8}{C(2 -c_1 - c_2)} + \frac{8}{C\sqrt{c_1 (2 - c_1 - c_2)}} +  \frac{4}{C^2(2-c_1-c_2)^2} + \frac{4}{C^2c_1(2 - c_1 - c_2)} \right) \E[Z']^2 \\
&= O(\E[Z']^2).\qedhere
\end{align*}
\end{itemize}
\end{proof}

To conclude the proof, we consider the two cases.
\begin{itemize}
\item Suppose $\dxs(p,q) \leq \ve^2$. 
By Lemma~\ref{lem:means} and the definition of $\mathcal{A}$, we have that $\E[Z] \leq (1 + c_1)m_2\ve^2$. 
By Lemma~\ref{lem:vars}, $\Var[Z] = O(m_2^2\ve^4)$.
Therefore, for constant $C$ sufficiently large, Chebyshev's inequality implies $\Pr(Z > \frac32 m_2 \ve^2) \leq 1/10$.
\item
Suppose $\dh(p,q) \geq \ve$.
By Lemma~\ref{lem:means}, we have that $\E[Z'] \geq (2 - c_1 - c_2)m_2\ve^2$. 
By Lemma~\ref{lem:vars}, $\Var[Z'] = O(\E[Z']^2)$.
Therefore, for constant $C$ sufficiently large, Chebyshev's inequality implies $\Pr(Z' < \frac32 m_2 \ve^2) \leq 1/10$.
Since $Z \geq Z'$, $\Pr(Z < \frac32 m_2 \ve^2) \leq 1/10$ as well.
\end{itemize}

\subsection{Identity Testing with $\ell_2$ Tolerance}
\label{sec:id-lt}
In this section, we sketch the algorithms required to achieve $\ell_2$ tolerance for identity testing.
Since the algorithms and analysis are very similar to those of Algorithm 1 of~\cite{AcharyaDK15} and Algorithm~\ref{alg:testing}, the full details are omitted.

First, we prove Theorem~\ref{thm:ones-tv}.
The algorithm is Algorithm 1 of~\cite{AcharyaDK15}, but instead of testing on $p$ and $q$, we instead test on $p^{+\frac12}$ and $q^{+\frac12}$, as defined in Proposition~\ref{prop:mixing}.
By this proposition, this operation preserves total variation and $\ell_2$ distance, up to a factor of $2$, and also makes it so that the minimum probability element of $q^{+\frac12}$ is at least $1/2n$. 
In the case where $\dlt(p,q) \leq \frac{\ve}{\sqrt{n}}$, we have the following upper bound on $\E[Z]$:
$$\E[Z'] = m \sum_{i \in \mathcal{\bar A}} \frac{(p_i - q_i)^2}{q_i} \leq O\left( m \cdot n \cdot \dlt^2(p,q)\right) \leq O(m_2 \ve^2).$$
This is the same bound as in Lemma 2 of~\cite{AcharyaDK15}.
The rest of the analysis follows identically to that of Algorithm 1 of~\cite{AcharyaDK15}, giving us Theorem~\ref{thm:ones-tv}.

Next, we prove Theorem~\ref{thm:ones-h}.
We observe that Algorithm~\ref{alg:testing} as stated can be considered as $\ell_2$-tolerant instead of $\chi^2$-tolerant, if desired.
First, we do not wrongfully reject any $p$ (i.e., those with $\dlt(p,q) \leq \frac{\ve^2}{\sqrt{n}}$) in Step~\ref{ln:early-reject}.
This is because we reject in this step if there is $\geq \Omega(\ve^2)$ total variation distance between $p$ and $q$ (witnessed by the set $\mathcal{\bar A}$), which implies that $p$ and $q$ are far in $\ell_2$-distance by Proposition~\ref{prop:ltinequalities}.
It remains to prove an upper bound on $\E[Z']$ in the case where $\dlt(p,q) \leq \frac{\ve^2}{\sqrt{n}}$.
$$\E[Z'] = m_2 \dxs(p,q) = m_2 \sum_{i \in \mathcal{\bar A}} \frac{(p_i - q_i)^2}{q_i} \leq O\left( m_2 \cdot \left(\frac{n}{\ve^2}\right)\cdot \dlt^2(p,q)\right) \leq O(m_2 \ve^2).$$
We note that this is the same bound as in Lemma~\ref{lem:means}.
With this bound on the mean, the rest of the analysis is identical to that of Theorem~\ref{thm:ones-csq-h}, giving us Theorem~\ref{thm:ones-h}.

%% file: twos-ub.tex
\section{Upper Bounds for Equivalence Testing}
\label{sec:twos-ub}
In this section, we prove the following theorems for equivalence testing.

\begin{theorem}\label{thm:twos-tv}
There exists an algorithm for equivalence testing between $p$ and $q$ distinguishing the cases:
\begin{itemize}
\item $\dlt(p,q) \leq \frac{\ve}{2\sqrt{n}}$ 
\item $\dtv(p,q) \geq \ve$
\end{itemize}
The algorithm uses $O\left(\max\left\{\frac{n^{2/3}}{\ve^{4/3}}, \frac{n^{1/2}}{\ve^2}\right\}\right)$ samples.
\end{theorem}

\begin{theorem}\label{thm:twos-h}
There exists an algorithm for equivalence testing between $p$ and $q$ distinguishing the cases:
\begin{itemize}
\item $\dlt(p,q) \leq \frac{\ve^2}{32\sqrt{n}}$ 
\item $\dh(p,q) \geq \ve$
\end{itemize}
The algorithm uses $O\left(\min\left\{\frac{n^{2/3}}{\ve^{8/3}}, \frac{n^{3/4}}{\ve^2}\right\}\right)$ samples.
\end{theorem}

Consider drawing~$\mathrm{Poisson}(m)$ samples from two unknown distributions $p = (p_1, \ldots, p_n)$ and $q = (q_1, \ldots, q_n)$.
Given the resulting histograms~$\bX$ and~$\bY$, \cite{ChanDVV14} define the following statistic:
\begin{equation}\label{eq:hel-statistic}
\bZ = \sum_{i=1}^n \frac{(\bX_i - \bY_i)^2 - \bX_i - \bY_i}{\bX_i + \bY_i}.
\end{equation}
This can be viewed as a modification to the empirical triangle distance applied to~$\bX$ and~$\bY$.
Both of our equivalence testing upper bounds will be obtained by appropriate thresholding of the statistic $\bZ$.

The organization of this section is as follows.
In Section~\ref{sec:eq-prelim}, we prove some basic properties of $\bZ$.
In Section~\ref{sec:eq-tv}, we prove Theorem~\ref{thm:twos-tv}.
In Section~\ref{sec:eq-h}, we prove Theorem~\ref{thm:twos-h}.

\subsection{Some facts about $\mathbf Z$}
\label{sec:eq-prelim}
Chan et al.~\cite{ChanDVV14} give the following expressions for the mean and variance of~$\bZ$.

\begin{proposition}[\cite{ChanDVV14}]\label{prop:mean-var}
Consider the function
\begin{equation*}
f(x) = \left(1 - \frac{1- e^{-x}}{x}\right).
\end{equation*}
Then for any subset $A\subseteq [n]$,
\begin{equation}\label{eq:z-mean}
\E[\bZ_A] = \sum_{i \in A} \frac{(p_i-q_i)^2}{p_i+q_i} m \cdot f(m(p_i + q_i)).
\end{equation}
As a result, $\bZ$ is mean-zero when $p= q$.  Furthermore,
\begin{equation*}
\Var[\bZ] \leq 2 \min\{m, n\} + \sum_{i=1}^n 5m \frac{(p_i - q_i)^2}{p_i + q_i}.
\end{equation*}
\end{proposition}
\noindent
Applying Proposition~\ref{prp:didnt-know-it-was-hellinger}, we immediately have the following corollary.
\begin{corollary}\label{cor:var-hel}
$\displaystyle
\Var[\bZ] \leq 2 \min\{m, n\} +  20m \dh(p,q)^2.
$
\end{corollary}

Without the corrective factor of $f(m(p_i + q_i))$,
Equation~\eqref{eq:z-mean} would just be~$m$ times the triangle distance between~$p$ and~$q$.
Our goal then is to understand the function $f(x)$ and how it affects this quantity.
Aside from the removable discontinuity at $x=0$, $f$ is a monotonically increasing function,
and for $x > 0$, it is strictly bounded between~$0$ and~$1$.
Furthermore, for~$x > 0$ there are roughly two ``regimes" that $f(x)$ exhibits:
when $x < 1$, where $f(x)$ is well-approximated by $x/2$,
and when $x \geq 1$, where $f(x)$ is ``morally the constant one,'' slowly increasing from~$e^{-1}$ to~$1$.
In fact, we have the following explicit bound on~$f(x)$.
\begin{fact}\label{fact:f-upper}
For all $x > 0$, $f(x) \leq \min\{1, x\}.$
\end{fact}
\noindent
In terms of $f(m(p_i + q_i))$, these regimes correspond to whether $p_i + q_i$ is less than or greater than~$\frac{1}{m}$.
Hence, the expression for the mean of~$\bZ$ (i.e.\ Equation~\eqref{eq:z-mean} for $A = [n]$)
splits in two: those terms for ``large" $p_i + q_i$ look roughly like the triangle distance (times~$m$),
and those terms for ``small" $p_i + q_i$ look roughly like the $\ell_2^2$ distance (times~$m^2$).
This is why we have given ourselves the flexibility to consider subsets~$A$ of the domain.

We will now prove several upper and lower bounds on $\E[\bZ_A]$,
based in part on whether we will apply them in the large or small $p_i + q_i$ regime.
Let us begin with a pair of upper bounds.

\begin{proposition}\label{prop:well-conditioned-upper}
Suppose for every $i \in A$, $p_i + q_i \geq \delta$. Then
\begin{equation*}
\E[\bZ_A] \leq \frac{m}{\delta} \dlt^2(p_A, q_A).
\end{equation*}
\end{proposition}
\begin{proof}
Because $f(x) \leq 1$ for all $x > 0$,
\begin{equation*}
\E[\bZ_A]
= \sum_{i \in A} \frac{(p_i-q_i)^2}{p_i+q_i} m \cdot f(m(p_i + q_i))
\leq  \sum_{i \in A} \frac{(p_i-q_i)^2}{p_i+q_i} m
\leq \frac{m}{\delta} \sum_{i \in A} (p_i-q_i)^2
=  \frac{m}{\delta} \dlt^2(p_A, q_A).\qedhere
\end{equation*}
\end{proof}

\begin{proposition}\label{prop:general-upper}
$\displaystyle
\E[\bZ] \leq m^2 \dlt^2(p, q).
$
\end{proposition}
\begin{proof}
Let $L$ be the set of $i$ such that $m(p_i + q_i) \geq 1$.
Then $\E[\bZ] = \E[\bZ_L] + \E[\bZ_{\overline{L}}]$, and by Proposition~\ref{prop:well-conditioned-upper},
$\E[\bZ_L] \leq m^2 \dlt^2(p_L, q_L)$.
On the other hand, by Fact~\ref{fact:f-upper}, $f(x) \leq x$, and therefore
\begin{equation*}
\E[\bZ_{\overline{L}}]
= \sum_{i \in \overline{L}} \frac{(p_i-q_i)^2}{p_i+q_i} m \cdot f(m(p_i + q_i))
\leq \sum_{i \in \overline{L}} (p_i-q_i)^2 m^2
= m^2 \dlt^2(p_{\overline{L}}, q_{\overline{L}}).
\end{equation*}
The proof is completed by noting that $\dlt^2(p_L, q_L) + \dlt^2(p_{\overline{L}}, q_{\overline{L}}) = \dlt^2(p, q).$
\end{proof}

Now we give a pair of lower bounds.

\begin{proposition}\label{prop:pretty-much-a-trivial-lower-bound-i-dont-know-what-to-tell-you}
Suppose for every $i \in A$, $m(p_i + q_i) \geq 1$.  Then
\begin{equation*}
\E[\bZ_A] \geq \frac{2m}{3} \dh^2(p_A,q_A).
\end{equation*}
\end{proposition}
\begin{proof}
Because $f(x)$ is monotonically increasing and $f(1) = 1/e$,
\begin{equation*}
\E[\bZ_A]
= m \sum_{i \in A} \frac{(p_i-q_i)^2}{p_i+q_i} f(m(p_i+q_i))
\geq m \sum_{i \in A} \frac{(p_i-q_i)^2}{p_i+q_i} f(1)
\geq \frac{2m}{e}\dh^2(p_A,q_A),
\end{equation*}
where the first step is by Proposition~\ref{prop:mean-var} and the last is by Proposition~\ref{prp:didnt-know-it-was-hellinger}.
The result follows from $e \leq 3$.
\end{proof}

The next proposition is essentially the second half of the proof of Lemma~$4$ from~\cite{ChanDVV14}.

\begin{proposition}\label{prop:their-bound}
For any subset~$A$,
\begin{equation*}
\E[\bZ_A] \geq \left(\frac{4m^2}{2|A| + m\cdot(p(A) + q(A))}\right)\cdot \dtv^2(p_A,q_A),
\end{equation*}
where we write $p(A) = \sum_{i \in A} p(i)$ and likewise for~$q(A)$.
\end{proposition}
\begin{proof}
Consider the function $g(x) = x f(x)^{-1}$.
Then $g(x) \leq 2+x$ for nonnegative~$x$.
Furthermore,
\begin{equation*}
\frac{(p_i - q_i)^2}{g(m(p_i + q_i))} = \frac{(p_i-q_i)^2}{m(p_i+q_i)} \left(1 - \frac{1-e^{-m(p_i+q_i)}}{m(p_i + q_i)}\right),
\end{equation*}
which, from Proposition~\ref{prop:mean-var}, is $\frac{1}{m^2} \cdot \E[\bZ_{\{i\}}]$.
As a result,
\begin{multline*}
\dtv^2(p_A,q_A)
= \frac14 \left(\sum_{i \in A} |p_i - q_i|\right)^2
= \frac14 \left(\sum_{i \in A} |p_i - q_i|\cdot \frac{\sqrt{g(m(p_i+q_i))}}{\sqrt{g(m(p_i+q_i))}}\right)^2\\
\leq \frac14\left(\sum_{i \in A} \frac{(p_i - q_i)^2}{g(m(p_i + q_i))}\right) \cdot \left(\sum_{i \in A} g(m(p_i + q_i))\right)
\leq \frac{1}{4m^2} \cdot \E[\bZ_A] \cdot(2|A| + m\cdot (p(A) + q(A))),
\end{multline*}
where the first inequality is Cauchy-Schwarz. Rearranging finishes the proof.
\end{proof}

\subsection{Equivalence Testing with Total Variation Distance}
\label{sec:eq-tv}
In this section, we prove Theorem~\ref{thm:twos-tv}. 
We will take the number of samples to be 
\begin{equation}\label{eq:l1-max}
m = \max\left\{C\cdot \frac{n^{2/3}}{\epsilon^{4/3}}, C^{3/2}\cdot \frac{n^{1/2}}{\epsilon^2}\right\},
\end{equation}
where~$C$ is some constant which can be taken to be~$10^{10}$.
 
Rather than drawing samples from~$p$ or~$q$,
our algorithm draws samples from~$p^{+1/2}$ and~$q^{+1/2}$.
By Proposition~\ref{prop:mixing}, we have the following guarantees in the two cases:
\begin{equation*}
\text{(Case 1):}~\dlt(p^{+1/2},q^{+1/2}) \leq \frac{\epsilon}{4 \sqrt{n}},
\qquad
\text{(Case 2):}~\dtv(p^{+1/2},q^{+1/2}) \geq \frac{\epsilon}{2}.
\end{equation*}
Furthermore, for any $i \in [n]$, we know the $i$-th coordinates of $p^{+1/2}$ and~$q^{+1/2}$ are both at least~$\frac{1}{2n}$.
Henceforth, we will write~$p'$ and~$q'$ for~$p^{+1/2}$ and~$q^{+1/2}$, respectively.

In Case 1, if we apply Proposition~\ref{prop:well-conditioned-upper} with $A = [n]$ and $\delta = \frac{1}{n}$
and Proposition~\ref{prop:general-upper},
\begin{equation*}
\E[\bZ]
\leq \min\{m^2, mn\} \cdot \dlt^2(p',q')
\leq \min\{m^2, mn\} \cdot \frac{\epsilon^2}{16 n}
\leq \frac{m^2}{4(2m + 2n)} \cdot \epsilon^2.
\end{equation*}
On the other hand, in Case 2, applying Proposition~\ref{prop:their-bound} with $A = [n]$,
\begin{equation*}
\E[\bZ] \geq \frac{4m^2}{2m + 2n} \cdot \dtv(p',q')^2 \geq \frac{m^2}{2m + 2n} \cdot \epsilon^2.
\end{equation*}
Our algorithm therefore thresholds~$\bZ$ on the value $\frac{5 m^2}{8(2m +2n)} \epsilon^2$,
outputting ``close" if it's below this value and ``far" otherwise.

The two bounds in~\eqref{eq:l1-max} meet when $C^3 \epsilon^{-4} = n$,
which is exactly when $m = n$.
When $m \leq n$, the first bound applies, and when $m > n$ the second bound applies.
As a result, we will split our analysis into the two cases.

\begin{lemma}
The tester succeeds in the $m \leq n$ case of Theorem~\ref{thm:twos-tv}.
\end{lemma}
\begin{proof}
By Corollary~\ref{cor:var-hel}
\begin{equation*}
\Var[\bZ] \leq 2 \min\{m, n\} +  20m \dh(p',q')^2
\leq 22m,
\end{equation*}
where we used the fact that $\dh(p',q') \leq 1$.
In Case 1, by Chebyshev's inequality,
\begin{equation*}
\Pr\left[\bZ \geq \frac{5 m^2}{8(2m +2n)} \epsilon^2\right]
\leq \frac{\Var[\bZ]}{\left(\frac{3m^2}{8(2m +2n)} \epsilon^2\right)^2}
= O\left(\frac{m}{\frac{m^4}{n^2} \epsilon^4}\right)
= O\left(\frac{n^2}{m^3 \epsilon^4}\right).
\end{equation*}
In Case 2,
\begin{equation*}
\Pr\left[\bZ \leq \frac{5 m^2}{8(2m +2n)} \epsilon^2\right]
\leq \frac{64 \Var[\bZ]}{9\E[\bZ]^2}
= O\left(\frac{m}{\frac{m^4}{n^2} \epsilon^4}\right)
= O\left(\frac{n^2}{m^3 \epsilon^4}\right).
\end{equation*}
Both of these bounds can be made arbitrarily small constants by setting~$C$ sufficiently large.
\end{proof}

\begin{lemma}
The tester succeeds in the $m \geq n$ case of Theorem~\ref{thm:twos-tv}.
\end{lemma}
\begin{proof}
We first consider Case 1.
By Proposition~\ref{prop:mean-var},
\begin{equation*}
\Var[\bZ]
\leq 2 \min\{m, n\} + \sum_{i=1}^n 5m \frac{(p_i' - q_i')^2}{p_i' + q_i'}
\leq 2 n + 5 m n \dlt^2(p',q')
\leq 2 n + \tfrac{5}{16} m \epsilon^2.
\end{equation*}
Then, we have that
\begin{equation*}
\Pr\left[\bZ \geq \frac{5 m^2}{8(2m +2n)} \epsilon^2\right]
\leq \frac{\Var[\bZ]}{\left(\frac{3m^2}{8(2m +2n)} \epsilon^2\right)^2}
= O\left(\frac{n}{m^2 \epsilon^4} + \frac{m\epsilon^2}{m^2 \epsilon^4}\right)
= O\left(\frac{n}{m^2 \epsilon^4} + \frac{1}{m \epsilon^2}\right).
\end{equation*}
Next, we focus on Case 2.
Write $L$ for the set of $i \in [n]$ such that $m(p_i' + q_i') \geq 1$.
Then $\dh^2(p_{\overline{L}}',q_{\overline{L}}') \leq \frac12 \sum_{i \in \overline{L}} (p_i' + q_i') \leq n/2m$.
As a result, 
by Corollary~\ref{cor:var-hel}
\begin{equation*}
\Var[\bZ] \leq 2 \min\{m, n\} +  20m \dh^2(p',q')
\leq 12 n + 20m \dh^2(p_L',q_L'). 
\end{equation*}
By Proposition~\ref{prop:pretty-much-a-trivial-lower-bound-i-dont-know-what-to-tell-you},
$\E[\bZ] \geq \frac{2m}{3} \dh^2(p_L',q_L')$.
Hence, 
\begin{align*}
\Pr\left[\bZ \leq \frac{5 m^2}{8(2m +2n)} \epsilon^2\right]
&\leq \frac{64 \Var[\bZ]}{9\E[\bZ]^2}
= O\left(\frac{n}{\E[\bZ]^2} + \frac{m \dh^2(p_L',q_L')}{\E[\bZ]^2}\right)\\
&= O\left(\frac{n}{\E[\bZ]^2} + \frac{1}{\E[\bZ]}\right)
= O\left(\frac{n}{m^2 \epsilon^4} + \frac{1}{m \epsilon^2}\right).
\end{align*}
Both of these bounds can be made arbitrarily small constants by setting~$C$ sufficiently large.
\end{proof}

\subsection{Equivalence Testing with Hellinger Distance}
\label{sec:eq-h}

In this section, we prove Theorem~\ref{thm:twos-h}. 
We will take the number of samples to be 
\begin{equation*}
m = \min\left\{C\cdot \frac{n^{2/3}}{\epsilon^{8/3}}, C^{3/4}\cdot \frac{n^{3/4}}{\epsilon^2}\right\},
\end{equation*}
where~$C$ is some constant which can be taken to be~$10^{10}$.

Rather than drawing samples from~$p$ or~$q$,
our algorithm draws samples from~$p^{+\delta}$ and~$q^{+\delta}$
for $\delta = \epsilon^2/32$.
By Proposition~\ref{prop:mixing}, we have the following guarantees in the two cases:
\begin{equation*}
\text{(Case 1):}~\dlt(p,q) \leq \frac{\epsilon^2}{32 \sqrt{n}},
\qquad
\text{(Case 2):}~\dh(p,q) \geq \frac{1}{2} \epsilon.
\end{equation*}
Furthermore, for any $i \in [n]$, we know the $i$-th coordinates of $p^{+\delta}$ and~$q^{+\delta}$ are both at least~$\frac{\epsilon^2}{32n}$.
Henceforth, we will write~$p'$ and~$q'$ for~$p^{+\delta}$ and~$q^{+\delta}$, respectively.

The two bounds meet when $C^{3/4}\epsilon^{-2} = n^{1/4}$,
which is exactly when $m = n$.
When $m \leq n$, the first bound applies, and when $m > n$ the second bound applies.
As a result, we will split our analysis into the two cases.

\begin{lemma}
The tester succeeds in the $m \leq n$ case of Theorem~\ref{thm:twos-h}.
\end{lemma}
\begin{proof}
In Case 1, if we apply Proposition~\ref{prop:general-upper},
\begin{equation*}
\E[\bZ]
\leq m^2 \cdot  \dlt^2(p',q')
\leq \frac{m^2 \epsilon^4}{32^2 n}.
\end{equation*}
On the other hand, in Case 2, applying Proposition~\ref{prop:their-bound} with $A = [n]$,
\begin{equation*}
\E[\bZ]
\geq \left(\frac{4m^2}{2n+2m}\right) \cdot \dtv(p',q')^2
\geq \left(\frac{4m^2}{2n+2m}\right) \cdot \dh(p',q')^4
\geq \frac{m^2\epsilon^4}{16n}.
\end{equation*}
Our algorithm therefore thresholds~$\bZ$ on the value $ \frac{m^2\epsilon^4}{128n}$,
outputting ``close" if it's below this value and ``far" otherwise.

By Corollary~\ref{cor:var-hel}
\begin{equation*}
\Var[\bZ] \leq 2 \min\{m, n\} +  20m \dh(p',q')^2
\leq 22m,
\end{equation*}
where we used the fact that $\dh(p',q') \leq 1$.
In Case 1,
\begin{equation*}
\Pr\left[\bZ \geq \frac{m^2\epsilon^4}{128n}\right]
\leq \frac{\Var[\bZ]}{\left(\frac{m^2\epsilon^4}{256n}\right)^2}
= O\left(\frac{m}{\frac{m^4}{n^2} \epsilon^8}\right)
= O\left(\frac{n^2}{m^3 \epsilon^8}\right).
\end{equation*}
In Case 2,
\begin{equation*}
\Pr\left[\bZ \leq \frac{m^2\epsilon^4}{128n}\right]
\leq \frac{64 \Var[\bZ]}{49\E[\bZ]^2}
= O\left(\frac{m}{\frac{m^4}{n^2} \epsilon^8}\right)
= O\left(\frac{n^2}{m^3 \epsilon^8}\right).
\end{equation*}
Both of these bounds can be made arbitrarily small constants by setting~$C$ sufficiently large.
\end{proof}

\begin{lemma}
The tester succeeds in the $m > n$ case of Theorem~\ref{thm:twos-h}.
\end{lemma}
\begin{proof}
In Case 1, if we apply Proposition~\ref{prop:well-conditioned-upper} with $A = [n]$ and $\delta = \frac{\epsilon^2}{16n}$
and Proposition~\ref{prop:general-upper},
\begin{equation*}
\E[\bZ]
\leq \min\left\{m^2, 16\frac{mn}{\epsilon^2}\right\} \cdot \dlt^2(p',q')
\leq \min\left\{m^2, 16\frac{mn}{\epsilon^2}\right\} \cdot \frac{\epsilon^4}{32^2 n}
= \min\left\{\frac{m^2\epsilon^4}{32^2 n}, \frac{m\epsilon^2}{64}\right\}.
\end{equation*}
Case 2 is more complicated.
We will need to define the set of ``large" coordinates
$L = \{i : m (p_i' + q_i') \geq 1\}$
and the set of ``small" coordinates $S = [n] \setminus L$.
Applying Proposition~\ref{prop:their-bound} to~$S$, we have
\begin{equation*}
\E[\bZ_S] \geq \left(\frac{4m^2}{2|S| + m\cdot(p'(S) + q'(S))}\right)\cdot \dtv^2(p_S',q_S')
\geq \frac{4m^2}{3n} \dtv^2(p_S',q_S'),
\end{equation*}
where $m\cdot(p'(S)+q'(S)) \leq n$ by the definition of~$S$.
If we also apply Proposition~\ref{prop:pretty-much-a-trivial-lower-bound-i-dont-know-what-to-tell-you} to~$L$,
we get
\begin{equation*}
\E[\bZ] = \E[\bZ_S] + \E[\bZ_L]
\geq \frac{4m^2}{3n} \dtv^2(p_S',q_S') + \frac{2m}{3} \dh^2(p_L',q_L')
\geq \min\left\{\frac{m^2\epsilon^4}{48n}, \frac{m\epsilon^2}{12}\right\},
\end{equation*}
where the last step follows
because $\dh^2(p_S',q_S') + \dh^2(p_L',q_L') = \dh^2(p',q')$ and $\dtv^2(p_S',q_S') \geq \dh^4(p_S',q_S')$.
As a result, we threshold~$\bZ$ on the value
\begin{equation*}
\frac{1}{2} \cdot \min\left\{\frac{m^2\epsilon^4}{48n}, \frac{m\epsilon^2}{12}\right\},
\end{equation*}
outputting ``close" if it's below this value and ``far" otherwise.

In Case 1, by Proposition~\ref{prop:mean-var},
\begin{equation*}
\Var[\bZ]
\leq 2 \min\{m, n\} + \sum_{i=1}^m 5m \frac{(p_i' - q_i')^2}{p_i' + q_i'}
\leq 2 n + \frac{80 m n}{\epsilon^2} \Vert p' - q' \Vert_2^2
\leq 2 n + \frac{5}{64}m \epsilon^2.
\end{equation*}
Hence, by Chebyshev's inequality,
\begin{multline*}
\Pr\left[\bZ \geq\frac{1}{2} \cdot \min\left\{\frac{m^2\epsilon^4}{48n}, \frac{m\epsilon^2}{12}\right\}\right]
\leq \frac{\Var[\bZ]}{(\frac{1}{8} \cdot \min\left\{\frac{m^2\epsilon^4}{48n}, \frac{m\epsilon^2}{12}\right\})^2}\\
\leq O\left(\frac{n}{(\frac{m^2 \epsilon^4}{n})^2} + \frac{n}{(m \epsilon^2)^2}
	+ \frac{m\epsilon^2}{(\frac{m^2 \epsilon^4}{n})^2} + \frac{m\epsilon^2}{(m\epsilon^2)^2}\right)\\
= O\left(\frac{n^3}{m^4 \epsilon^8} + \frac{n}{m^2 \epsilon^4}
	+ \frac{n^2}{m^3 \epsilon^6} + \frac{1}{m\epsilon^2}\right).
\end{multline*}
This can be made an arbitrarily small constant by setting~$C$ sufficiently large.

In Case 2, by Corollary~\ref{cor:var-hel},
\begin{equation}\label{eq:gonna-split}
\Pr\left[\bZ \leq \frac{\E[\bZ]}{2} \right]
\leq \frac{4 \Var[\bZ]}{\E[\bZ]^2}
\leq \frac{8 n + 80 m\dh(p',q')^2}{\E[\bZ]^2}.
\end{equation}
Because $\dh(p',q')^2 = \dh^2(p_S',q_S') + \dh^2(p_L',q_L')$,
either $\dh^2(p_S',q_S')$ or $\dh^2(p_L',q_L')$ is at least $\frac{1}{2}\dh^2(p',q')$.
Suppose that $\dh^2(p_S',q_S') \geq \frac{1}{2}\dh^2(p',q')$.
We note that
\begin{equation*}
m \dh^2(p_S',q_S')
= \frac{m}{2} \sum_{i \in S} (\sqrt{p_i'} - \sqrt{q_i'})^2
\leq \frac{m}{2} \sum_{i \in S} |p_i' + q_i'|
\leq \frac{n}{2},
\end{equation*}
by the definition of~$S$.
Thus,
\begin{equation*}
\eqref{eq:gonna-split}
\leq \frac{8n + 160 m \dh^2(p_S',q_S')}{(\frac{4m^2}{3n}\dtv^2(p_S',q_S'))^2}
\leq \frac{88n}{(\frac{4m^2}{3n}\dtv^2(p_S',q_S'))^2}
= O\left(\frac{n^3}{m^4 \dtv^4(p_S',q_S')}\right)
\leq O\left(\frac{n^3}{m^4 \epsilon^8}\right),
\end{equation*}
where the last step used the fact that $\dtv(p_S',q_S') \geq \dh^2(p_S',q_S') \geq \frac{1}{2}\dh^2(p',q') \geq \frac{1}{2}\epsilon^2$.

In the case when $\dh^2(p_L',q_L') \geq \frac{1}{2} \dh^2(p',q')$,
\begin{equation*}
\eqref{eq:gonna-split}
\leq \frac{8n + 160 m \dh^2(p_L',q_L')}{(\frac{2m}{3} \dh^2(p_L',q_L'))^2}
= O\left(\frac{n}{m^2 \dh^4(p_L',q_L')} + \frac{1}{m \dh^2(p_L',q_L')}\right)
\leq O\left(\frac{n}{m^2 \epsilon^4} + \frac{1}{m \epsilon^2}\right).
\end{equation*}
This can be made an arbitrarily small constant by setting~$C$ sufficiently large.
\end{proof}

%% file: est-ub.tex
\section{Upper Bounds Based on Estimation}
\label{sec:est-ub}

We start by showing a simple meta-algorithm -- in short, it says that if a testing problem is well-defined (i.e., has appropriate separation between the cases) and we can estimate one of the distances, it can be converted to a testing algorithm.
\begin{theorem}\label{thm:est-ub}
Suppose there exists an $m(n, \ve)$-sample algorithm which, given sample access to distributions $p$ and $q$ over $[n]$ and parameter $\ve$, estimates some distance $d(p,q)$ up to an additive $\ve$ with probability at least $2/3$.
Consider distances $d_X(\cdot, \cdot), d_Y(\cdot, \cdot)$ and $\ve_1, \ve_2 > 0$ such that $ d_Y(p,q) \geq \ve_2 \rightarrow d_X(p,q) > 3\ve_1/2$ and $d_X(p,q) \leq \ve_1 \rightarrow d_Y(p,q) < 2\ve_2/3$, and $d(\cdot, \cdot)$ is either $d_X(\cdot, \cdot)$ or $d_Y(\cdot, \cdot)$.

Then there exists an algorithm for equivalence testing between $p$ and $q$ distinguishing the cases:
\begin{itemize}
\item $d_X(p,q) \leq \ve_1$;
\item $d_Y(p,q) \geq \ve_2$.
\end{itemize}
The algorithm uses either $m(n, O(\ve_1))$ or $m(n, O(\ve_2))$ samples, depending on whether $d = d_X$ or $d_Y$.
\end{theorem}
\begin{proof}
Suppose that $d = d_X$, the other case follows similarly.
Using the $m(n, \ve_1/4)$ samples, obtain an estimate $\hat \tau$ of $d_X(p,q)$, accurate up to an additive $\ve_1/4$.
If $\hat \tau \leq 5\ve_1/4$, output that $d_X(p,q) \leq \ve_1$, else output that $d_Y(p,q) \geq \ve_2$. 
Conditioning on the correctness of the estimation algorithm, correctness for the case when $d_X(p,q) \leq \ve_1$ is immediate, and correctness for the case when $d_Y(p,q) \geq \ve_2$ follows from the separation between the cases.
\end{proof}

It is folklore that a distribution over $[n]$ can be $\ve$-learned in $\ell_2$-distance with $O(1/\ve^2)$ samples (see, i.e., \cite{ChanDVV14, Waggoner15} for a reference).
By triangle inequality, this implies that we can estimate the $\ell_2$ distance between $p$ and $q$ up to an additive $O(\ve)$ with $O(1/\ve^2)$ samples, leading to the following corollary.

\begin{corollary}\label{cor:l2-est}
There exists an algorithm for equivalence testing between $p$ and $q$ distinguishing the cases:
\begin{itemize}
\item $d(p,q) \leq f(n, \ve)$;
\item $\dlt(p,q) \geq \ve$,
\end{itemize}
where $d(\cdot, \cdot)$ is a distance and $f(n, \ve)$ is such that $\dlt(p,q) \geq \ve \rightarrow d(p,q) \geq 3f(n, \ve)/2$ and $d(p,q) \leq f(n, \ve) \rightarrow \dlt(p,q) \leq 2\ve/3$.  
The algorithm uses $O(1/\ve^2)$ samples.
\end{corollary}

Finally, we note that total variation distance between $p$ and $q$ can be additively estimated up to a constant using $O(n/\log n)$ samples~\cite{LehmannC06, ValiantV11b, JiaoHW16}, leading to the following corollary:
\begin{corollary}\label{cor:tv-est}
For constant $\ve > 0$, there exists an algorithm for equivalence testing between $p$ and $q$ distinguishing the cases:
\begin{itemize}
\item $\dtv(p,q) \leq \ve^2/4$;
\item $\dh(p,q) \geq \ve/\sqrt{2}$.
\end{itemize}
The algorithm uses $O(n/\log n)$ samples.
\end{corollary}

%% file: lb.tex
\section{Lower Bounds}
\label{sec:lb}
We start with a simple lower bound, showing that identity testing with respect to KL divergence is impossible.
A similar observation was made in~\cite{BatuFRSW00}.
\begin{theorem}\label{thm:untestable}
No finite sample test can perform identity testing between $p$ and $q$ distinguishing the cases:
\begin{itemize}
\item $p = q$;
\item $\dkl(p,q) \geq \ve^2$.
\end{itemize}
\end{theorem}
\begin{proof}
Simply take $q = (1, 0)$ and let~$p$ be either $(1, 0)$ or $(1-\delta, \delta)$, for~$\delta > 0$ tending to zero.
Then $p = q$ in the first case and $\dkl(p,q) = \infty$ in the second, but distinguishing between these two possibilities for~$p$
takes $\Omega(\delta^{-1})\rightarrow \infty$ samples.
\end{proof}

Next, we prove our lower bound for KL tolerant identity testing.

\begin{theorem}\label{thm:ones-lb}
There exist constants $0 < s < c$, such that any algorithm for identity testing between $p$ and $q$ distinguishing the cases:
\begin{itemize}
\item $\dkl(p,q) \leq s$;
\item $\dtv(p,q) \geq c$;
\end{itemize}
requires $\Omega(n/\log n)$ samples.
\end{theorem}
\begin{proof}
Let $q = (\tfrac{1}{n}, \ldots, \tfrac{1}{n})$ be the uniform distribution.
Let $R(\cdot, \cdot)$ denote the \emph{relative earthmover distance} (see~\cite{ValiantV10a} for the definition).
By Theorem~$1$ of~\cite{ValiantV10a},
for any $\delta < \frac{1}{4}$
there exist sets of distributions~$\calC$ and~$\calF$ (for \emph{close} and \emph{far})
such that:
\begin{itemize}
\item For every $p \in \calC$, $R(p, q) = O(\delta | \log \delta|)$.
\item For every $p \in \calF$ there exists a distribution~$r$ which is uniform over~$n/2$ elements such that $R(p, r) = O(\delta | \log \delta|)$.
\item Distinguishing between $p \in \calC$ and $p \in \calF$ requires $\Omega(\frac{\delta n}{\log(n)})$ samples.
\end{itemize}
Now, if $p \in \calC$ then
\begin{equation*}
\dkl(p,q)
= \sum_{i=1}^n p_i \log\left(\frac{p_i}{1/n}\right)
= \log(n) - H(p)
\leq O(\delta |\log(\delta)|),
\end{equation*}
where $H(p)$ is the Shannon entropy of~$p$,
and here we used the fact that $|H(p) - H(q)| \leq R(p, q)$, which follows from Fact~$5$ of~\cite{ValiantV10a}.
On the other hand, if $q \in \calF$, let~$r$ be the corresponding distribution which is uniform over~$n/2$ elements.
Then
\begin{equation*}
\frac{1}{2}
= \dtv(p,q)
\leq \dtv(q,p) + \dtv(p,r)
\leq \dtv(q,p) + O(\delta | \log \delta|),
\end{equation*}
where we used the triangle inequality
and the fact that $\dtv(p,r) \leq R(p, r)$ (see~\cite{ValiantV10a} page 4).
As a result, if we set~$\delta$ to be some small constant,
$s = O(\delta |\log(\delta)|)$,
and $c = \frac{1}{2} - O(\delta | \log\delta|)$,
then this argument shows that distinguish $\dkl(p,q) \leq s$ versus $\dtv(p,q) \geq c$
requires $\Omega(n/\log n)$ samples.
\end{proof}

Finally, we conclude with our lower bound for $\chi^2$-tolerant equivalence testing.

\begin{theorem}\label{thm:twos-lb}
There exists a constant $\ve > 0$ such that any algorithm for equivalence testing between $p$ and $q$ distinguishing the cases:
\begin{itemize}
\item $\dxs(p,q) \leq \ve^2/4$;
\item $\dtv(p,q) \geq \ve$;
\end{itemize}
requires $\Omega(n/\log n)$ samples.
\end{theorem}
\begin{proof}
We reduce the problem of distinguishing $\dh(p,q) \leq \frac{1}{\sqrt{48}} \epsilon$ from $\dtv(p,q) \geq 3\epsilon$ to this.
Define the distributions
\begin{equation*}
p' = \frac{2}{3} p + \frac{1}{3} q, \qquad q' = \frac{1}{3} p + \frac{2}{3} q.
\end{equation*}
Then $m$ samples from~$p'$ and~$q'$ can be simulated by $m$ samples from~$p$ and~$q$.
Furthermore,
\begin{equation*}
\dh(p',q') \leq \frac{1}{\sqrt{48}} \epsilon, \qquad \dtv(p',q') = \frac{1}{3} \dtv(p,q) \geq \epsilon,
\end{equation*}
where we used the fact that Hellinger distance satisfies the data processing inequality.
But then, in the ``close" case,
\begin{equation*}
\dxs(p',q')
= \sum_{i=1}^n \frac{(p'_i - q'_i)^2}{q'_i}
\leq 3 \sum_{i=1}^n \frac{(p'_i - q'_i)^2}{p'_i + q'_i}
\leq 12 \dh^2(p',q') \leq \frac{1}{4} \epsilon^2,
\end{equation*}
where we used the fact that $p'_i \leq 2q'_i$ and Proposition~\ref{prp:didnt-know-it-was-hellinger}.
Hence, this problem, which requires $\Omega(n/\log n)$ samples (by the relationship between total variation and Hellinger distance, and the lower bound for testing total variation-close versus -far of~\cite{ValiantV10a}), reduces to the problem in the proposition, and so that requires $\Omega(n/\log n)$ samples as well.
\end{proof}

%% file: distanceinequalities.tex
\section{Proof of Proposition~\ref{prop:distanceinequalities}}
\label{sec:distanceinequalities}
Recall that we will prove this for restrictions of probability distributions to subsets of the support -- in other words, we do not assume $\sum_{i \in S} p_i = \sum_{i \in S} q_i = 1$, we only assume that $\sum_{i \in S} p_i \leq 1$ and $\sum_{i \in S} q_i \leq 1$.
\paragraph{$\dh^2(p_S,q_S) \leq \dtv(p_S,q_S):$}
\begin{align*}
\dh^2(p_S,q_S) &= \frac12 \sum_{i \in S} (\sqrt{p_i} - \sqrt{q_i})^2 \\
&\leq \frac12 \sum_{i \in S} |\sqrt{p_i} - \sqrt{q_i}|(\sqrt{p_i} + \sqrt{q_i}) \\
&= \frac12 \sum_{i \in S} |p_i - q_i| \\
&= \dtv(p_S,q_S).
\end{align*}

\paragraph{$\dtv(p_S,q_S) \leq \sqrt{2}\dh(p_S,q_S):$}
\begin{align*}
\dtv^2(p_S,q_S) &= \frac14 \left(\sum_{i \in S} \left|p_i - q_i\right|\right)^2 \\
&= \frac14 \left(\sum_{i \in S} \left|\sqrt{p_i} -\sqrt{q_i}\right|(\sqrt{p_i} + \sqrt{q_i})\right)^2 \\
&\leq \frac14 \left(\sum_{i \in S} \left|\sqrt{p_i} -\sqrt{q_i}\right|^2\right)\left(\sum_{i \in S}(\sqrt{p_i} + \sqrt{q_i})^2\right) \\
&\leq \dh^2(p_S, q_S) \cdot \frac12 \left(\sum_{i \in S}(\sqrt{p_i} + \sqrt{q_i})^2\right) \\
&= \dh^2(p_S, q_S) \cdot \left(\sum_{i \in S}p_i + \sum_{i \in S} q_i - \dh^2(p_S,q_S)\right) \\
&\leq \dh^2(p_S,q_S) \cdot \left(2 - \dh^2(p_S,q_S)\right) \\
&\leq 2 \dh^2(p_S,q_S).
\end{align*}
Taking the square root of both sides gives the result.
The second inequality is Cauchy-Schwarz.

\paragraph{$2\dh^2(p_S, q_S) \leq \sum_{i \in S} (q_i - p_i) + \dkl(p_S, q_S):$}
\begin{align*}
2 \dh^2(p_S, q_S) 
&= \sum_{i \in S} (q_i + p_i) - 2\sum_{i \in S} \sqrt{p_i q_i} \\
&= \sum_{i \in S} (q_i + p_i) - 2\left(\left(\sum_{j \in S} p_j\right)\sum_{i \in S} \frac{p_i}{\sum_{j \in S} p_j} \sqrt{\frac{q_i}{p_i}}\right) \\
&\leq \sum_{i \in S} (q_i + p_i) - 2\left(\left(\sum_{j \in S} p_j\right)\exp\left(\frac12 \sum_{i \in S} \frac{p_i}{\sum_{j \in S} p_j} \log{\frac{q_i}{p_i}}\right)\right) \\
&\leq \sum_{i \in S} (q_i + p_i) - 2\left(\left(\sum_{j \in S} p_j\right)\left(1 + \frac12 \sum_{i \in S} \frac{p_i}{\sum_{j \in S} p_j} \log{\frac{q_i}{p_i}}\right)\right) \\
&= \sum_{i \in S} (q_i - p_i) - \left(\sum_{i \in S} p_i \log{\frac{q_i}{p_i}} \right)\\
&= \sum_{i \in S} (q_i - p_i) + \dkl(p_S, q_S).
\end{align*}
The first inequality is Jensen's, and the second is $1 + x \leq \exp(x)$.

\paragraph{$\dkl(p_S, q_S) \leq \sum_{i \in S} (p_i - q_i) +  \dxs(p_S, q_S):$}
\begin{align*}
\dkl(p_S, q_S) 
&= \left(\sum_{j \in S} p_j\right)\left(\sum_{i \in S} \frac{p_i}{\sum_{j \in S} p_j} \log{\frac{p_i}{q_i}}\right) \\
&\leq \left(\sum_{j \in S} p_j\right)\left(\log {\frac{1}{\sum_{j\in S} p_j}\sum_{i \in S} \frac{p_i^2}{q_i}} \right) \\
&= \left(\sum_{j \in S} p_j\right) \left(\log{ \left(\frac{1}{\sum_{j\in S} p_j} \left(\dxs(p_S, q_S) + 2\sum_{i \in S} p_i - \sum_{i \in S} q_i\right)\right)} \right) \\
&= \left(\sum_{j \in S} p_j\right) \left(\log{  \left(2 + \frac{1}{\sum_{j\in S} p_j} \left(\dxs(p_S, q_S)  - \sum_{i \in S} q_i\right)\right)} \right) \\
&\leq \left(\sum_{j \in S} p_j\right) \left(1 + \frac{1}{\sum_{j\in S} p_j} \left(\dxs(p_S, q_S)  - \sum_{i \in S} q_i\right)\right) \\
&=\sum_{i \in S} (p_i - q_i) +  \dxs(p_S, q_S).
\end{align*}
The first inequality is Jensen's, and the second is $1 + x \leq \exp(x)$.